%% file: main.tex
\documentclass[a4paper]{report}

\input{header}
\input{metadata}

\begin{document}

\input{titlepage}
\input{abstract}
\clearpage
\belowpdfbookmark{Table of Contents}{toc}
\tableofcontents

\input{1introduction}
\input{2preliminaries}
\input{3lassos}
\input{4templates}
\input{5constraints}
\input{6nl}
\input{7solving}
\input{8conclusion}

\input{errata}

\clearpage
\phantomsection
\addcontentsline{toc}{chapter}{Bibliography}
\bibliographystyle{abbrv}
\bibliography{references}



\end{document}

%% file: header.tex
\usepackage[ngerman,english]{babel}
\usepackage[utf8]{inputenc}
\usepackage{setspace}
\usepackage{hyperref}
\usepackage{enumerate}
\usepackage[labelfont=bf]{caption}

\usepackage{amsmath}
\usepackage{amssymb}
\usepackage{amsthm}

\usepackage[usenames,dvipsnames]{xcolor}

\usepackage{tikz}
\usetikzlibrary{automata,arrows}

\usepackage{listings}
\lstdefinelanguage{pseudo}{
	morekeywords={if, then, else, while, assume, assert},
	sensitive=false,
	morecomment=[l]{//},
	morecomment=[s]{/*}{*/},
	morestring=[b]",
}
\colorlet{myred}{red!70!black}
\colorlet{myblue}{blue!70!black}
\colorlet{mywhite}{white}

\usepackage[bottom]{footmisc} 

\newcommand{\linkcolor}{myblue}
\hypersetup{
  colorlinks=true,
  linkcolor=\linkcolor,
  citecolor=\linkcolor,
  urlcolor=\linkcolor
} 

\usepackage{aliascnt}

\newcommand{\mynewtheorem}[2]{
	\newaliascnt{#1}{dummy}
	\newtheorem{#1}[#1]{#2}
	\aliascntresetthe{#1}
	\expandafter\def\csname #1autorefname\endcsname{#2}
}

\AtBeginDocument{

}

\usepackage{etoolbox}
\makeatletter
\patchcmd{\l@chapter}{1.0em}{0.7em}{}{} 
\makeatother

\theoremstyle{plain}
\mynewtheorem{theorem}{Theorem}
\mynewtheorem{proposition}{Proposition}
\mynewtheorem{lemma}{Lemma}
\mynewtheorem{corollary}{Corollary}
\mynewtheorem{remark}{Remark}
\mynewtheorem{conjecture}{Conjecture}
\theoremstyle{definition}
\mynewtheorem{definition}{Definition}
\mynewtheorem{example}{Example}

\renewcommand{\emptyset}{\text{\Large \o}}

\newcommand{\stemt}{{\scriptstyle\mathrm{STEM}}} 
\newcommand{\loopt}{{\scriptstyle\mathrm{LOOP}}} 
\newcommand{\T}{{\scriptstyle\mathrm{T}}} 
\newcommand{\prog}{\mathbf{P}} 
\newcommand{\abovebelow}[2]{(^{#1}_{#2})} 
\newcommand{\tr}[1]{#1^T} 
\newcommand{\cwhite}{\tikz {
  \draw[fill=mywhite] (0,0) circle (0.1);
}}
\newcommand{\cred}{\tikz {
  \draw[fill=myred] (0,0) circle (0.1);
}}
\newcommand{\cblue}{\tikz {
  \draw[fill=myblue] (0,0) circle (0.1);
}}

\newcounter{TransformationStep}
\newcommand{\transformationStep}[1]{
	\stepcounter{TransformationStep}
	~

	\noindent
	{\bf Transformation \arabic{TransformationStep}: {#1}}
}

%% file: metadata.tex
\newcommand{\pdftitle}{Ranking Function Synthesis for Linear Lasso Programs}
\newcommand{\pdfauthor}{Jan Leike}
\newcommand{\pdfkeywords}{Lasso Programs, Termination, Termination
  Condition, Multiphase Ranking Function, Lexicographic Ranking
  Function, Non-linear Constraints, Farkas' Lemma}
\newcommand{\pdfsubject}{Static Analysis}
\newcommand{\pdflang}{en-US}

\hypersetup{
  pdftitle={\pdftitle},
  pdfauthor={\pdfauthor},
  pdfkeywords={\pdfkeywords},
  pdfsubject={\pdfsubject},
  pdflang={\pdflang}
}
\title{\pdftitle}
\author{\pdfauthor}
\date{\today}

%% file: titlepage.tex

\begin{titlepage}
\begin{center}

\begin{spacing}{1.5}
\textsc{\Large University of Freiburg} \\
\textsc{\large Department of Computer Science} \\
\textsc{\large Chair of Software Engineering} \\[5cm]
\end{spacing}

\textsc{\Huge
\begin{spacing}{1.1}
Ranking Function Synthesis
for \\
Linear Lasso Programs
\end{spacing}
}

\vspace{1cm}

\textsc{\Large Master's Thesis\footnote{This is an error-corrected version of the original thesis, updated last on \today.}} \\[12mm]
\textsc{\Large \pdfauthor} \\[5mm]
June 30, 2013

\vfill

\emph{Supervisor:} \\[0.2cm]
\large Prof. Dr. Andreas Podelski

\end{center}
\end{titlepage}

%% file: abstract.tex

\clearpage
\belowpdfbookmark{Abstract}{abstract}
\abstract{
The scope of this work is the constraint-based synthesis of
termination arguments for the restricted class of programs called
\emph{linear lasso programs}.  A termination argument consists of
a ranking function as well as a set of supporting invariants.

We extend existing methods in several ways.  First, we use Motzkin's
Transposition Theorem instead of Farkas' Lemma.  This allows us to
consider linear lasso programs that can additionally contain strict
inequalities.  Existing methods are restricted to non-strict
inequalities and equalities.

Second, we consider several kinds of ranking functions: affine-linear,
piecewise and lexicographic ranking functions.  Moreover, we present a
novel kind of ranking function called \emph{multiphase ranking
function} which proceeds through a fixed number of phases such that
for each phase, there is an affine-linear ranking function.  As an
abstraction to the synthesis of specific ranking functions, we
introduce the notion \emph{ranking function template}.  This enables
us to handle all ranking functions in a unified way.

Our method relies on non-linear algebraic constraint solving as a
subroutine which is known to scale poorly to large problems.  As a
mitigation we formalize an assessment of the difficulty of our
constraints and present an argument why they are of an easier kind
than general non-linear constraints.

We prove our method to be complete: if there is a termination argument
of the form specified by the given ranking function template with a
fixed number of affine-linear supporting invariants, then our method
will find a termination argument.

To our knowledge, the approach we propose is the most powerful
technique of synthesis-based discovery of termination arguments for
linear lasso programs and encompasses and enhances several methods
having been proposed thus far~\cite{BMS05linrank,HHLP13,PR04}.
}

%% file: 1introduction.tex

\clearpage
\chapter{Introduction}
\label{ch-introduction}

Software verification as a branch of computer science studies the
automatic derivation of correctness properties of computer programs.
Termination is the property that no infinite program execution is
possible.  In this work we focus on the automatic discovery of a
termination argument given a program of a specific form.  Whether a
given program terminates is undecidable according to the Halting
Problem.  Hence there is no algorithm that finds a termination
argument for \emph{every} terminating program.  Because of this, we
content ourselves with considering
\emph{linear lasso programs}.  Linear lasso programs consist of
a \emph{stem} followed by a \emph{loop}.  Stem and loop each are
boolean combinations of affine-linear constraints.  For an example,
see \autoref{fig-lasso-introduction}.

\begin{figure}[ht]
\begin{center}
\begin{minipage}{40mm}
\begin{lstlisting}
assume($y > 1$);
while ($q \geq 0$):
    $q$ := $q - y$;
    $y$ := $y + 1$;
\end{lstlisting}
\end{minipage}
\begin{minipage}{60mm}
\begin{align*}
\stemt(q, y) \equiv y > 1 \;&~ \\
\loopt(q, y, q', y') \equiv q \geq 0 \;&\land\; q' = q - y \\
&\land\; y' = y + 1
\end{align*}
\end{minipage}
\end{center}
\caption{
A linear lasso program given as program code (left) and its
translation as stem and loop transition in linear arithmetic (right).
}
\label{fig-lasso-introduction}
\end{figure}

Lasso programs usually do not occur as stand-alone programs; rather,
they are encountered when a finite representation of an infinite path
in a control flow graph is needed.  For example, in (potentially
spurious) counter-examples in termination
analysis~\cite{CPR06,HLNR10,KSTTW08,KSTW10}, non-termination
analysis~\cite{GHMRX08}, stability analysis~\cite{CFKP11,PW07}, or
cost analysis~\cite{AAGP11,GZ10}.

In this work we build constraints from the given program code, such
that a termination argument for this program can be computed via
constraint solving.  The method we propose is more powerful than any
other constraint-based synthesis of termination arguments for linear
lasso programs proposed thus far (see \autoref{sec-related-work} for
an assessment).

First, by using \hyperref[thm-motzkin]{Motzkin's Transposition
Theorem} instead of \hyperref[lem-farkas-affine]{Farkas' Lemma}, we
are able to handle lasso programs that contain both strict and
non-strict inequalities. (For example, the program
in \autoref{fig-lasso-introduction} contains the strict inequality $y
> 1$ in the stem and only non-strict inequalities in the loop
transition.)  Existing methods disallowed strict inequalities or
sometimes resorted to the workaround of replacing a strict inequality
$a > b$ by $a \geq b + 1$, which only works for integer domains.

Second, instead of focusing on one type of ranking function, we use
several templates for ranking functions (e.g., affine-linear or
lexicographic ranking functions).  For this, we introduce the notion
of a \emph{ranking function template} that enables formalization of
ranking functions of various kinds, including the aforementioned ones.
When given a linear lasso program, we can prove its termination by
trying many different kinds of templates available by repeating our
method for each of them.

Furthermore, we present a novel ranking function that we
call \emph{multiphase ranking function}.  This ranking function
proceeds through a fixed finite number of phases, before terminating.
Each phase is ranked by an affine-linear function; when this function
becomes non-positive, we transition to the next phase.  These
multiphase ranking functions can be seen to be orthogonal to
lexicographic ranking functions: if a program has a lexicographic
ranking function, it generally does not have a multiphase ranking
function, or vice versa.  We give various examples of programs that
have a multiphase ranking function.

\begin{figure}[ht]
\begin{center}
\begin{minipage}{4cm}
\begin{lstlisting}
while ($q \geq 0$):
    $q$ := $q - y$;
    $y$ := $y + 1$;
\end{lstlisting}
\end{minipage}
\begin{minipage}{3cm}
\vspace{-5mm}
\begin{align*}
f_1(q, y) &= 1 - y \\
f_2(q, y) &= q + 1
\end{align*}
\end{minipage}
\end{center}
\caption{
An execution of this linear lasso program can be split into two
phases: first $y$ increases until it is positive, then $q$ decreases
until the loop condition $q \geq 0$ is violated.  We can discover the
affine-linear functions $f_1$ and $f_2$.  Together, they form a
multiphase ranking function where $f_1$ corresponds to phase one and
$f_2$ corresponds to phase two.
}
\label{fig-multiphase-introduction}
\end{figure}

Our constraint-based synthesis method can be summarized as follows.
The input is a linear lasso program as well as a linear ranking
function template.  The template yields a formula, which we augment by
adding constraints for affine-linear inductive supporting invariants.
These invariants contain information from the program stem that may be
indispensable to the program's termination proof.  Next, five
equivalence transformations are applied to the constraints, the last
of which is given by Motzkin's Theorem.  The last transformation
removes any universal quantifiers.  The resulting constraints are then
passed to an SMT solver which checks them for satisfiability; a
positive result implies that the program terminates.  Furthermore, a
satisfying assignment will yield the supporting invariants and a
ranking function.  These form a termination argument for the given
linear lasso program and thus can be used by another
tool~\cite{AAGP11,CFKP11,GZ10,CPR06,GHMRX08,HLNR10,KSTTW08,KSTW10,PW07}.

In addition to being sound, our method is complete in the following
sense.  If there is a termination argument in form of a fixed number
of affine-linear supporting invariants and a ranking function of the
form specified by the given ranking function template, then our method
will discover a termination argument.  In other words, the existence
of a solution is never lost in the process of transforming the
constraints.

Our method applies to linear lasso programs of rational and real
variable domains.  While it is feasible to use it for integer domains,
it is not complete for integers.  The main reason is that Motzkin's
Theorem does not hold over the integers.  In fact, the discovery of
affine-linear ranking functions for lasso programs without stem is
already co-NP-complete~\cite{BG13}.

In contrast to some related methods~\cite{HHLP13,PR04}, which we
extend in this work, the constraints we generate are not linear, but
rather non-linear algebraic constraints.  Solving these constraints is
decidable, but requires exponential time and space~\cite{GV88}.  Much
progress on non-linear SMT solvers has been made and present-day
algorithms routinely solve non-linear constraints of various
sizes~\cite{JM12}.  \emph{Cylindrical algebraic decomposition} (CAD)
seems to be the most successful practical method in these endeavors.

We will argue that the constraints we generate generally are not as
wicked as non-linear constraints can possibly be.  We assess the
number of variables that need to be assigned to make the constraints
linear.  For this we introduce the notion of \emph{suitable colorings}
for ranking function templates.  This is a criterion that states which
of the Motzkin coefficients that occur in non-linear operations we can
eliminate from the final constraints.  Moreover, we provide several
other optimizations that reduce the number of these variables.
Additionally, for the CAD algorithm we exemplarily discuss why in
practical cases, we can find assignments for invariants in polynomial
time.

The contributions of this work can be summarized as follows.
\begin{itemize}
\item The use of Motzkin's Theorem instead of Farkas' Lemma enables
  strict inequalities in linear lasso programs.
\item The novel multiphase ranking function is presented.
\item We handle synthesis of different types of ranking functions in a
  unified way using our notion of ranking function templates.
\item The number of variables occurring in non-linear operations in
  the generated constraints is assessed for every ranking function
  template we present.
\item We argue why solving our non-linear constraints is not terribly
  difficult.
\end{itemize}

\section{Related Work}
\label{sec-related-work}

Tiwari showed that termination is decidable for deterministic
stem-free linear lasso programs of the form
\begin{align*}
\texttt{while(} Bx > b \texttt{) } x \texttt{:=} Ax + c \texttt{;}
\end{align*}
where $Bx > b$ is a conjunction of affine-linear constraints and $Ax +
c$ is an affine-linear transition function~\cite{Tiwari04}.  This
result is based on eigenvector analysis of the involved matrix $A$.
Braverman extends this result and proves the decidability of lasso
programs of the following form~\cite{Braverman06}:
\begin{align*}
\texttt{while(} B_s x > b_s \land B_w x \geq b_w \texttt{) }
  x \texttt{:=} Ax + c \texttt{;}
\end{align*}
where the matrices and vectors are rational and variables have
rational or real domain.  Moreover, this class of lasso programs also
admits decidable termination analysis over integer domain for the
homogeneous case where $b_s, b_w, c = 0$.

Ben-Amram et al.\ show that linear lasso programs with integer domain
have undecidable termination if the loop's coefficients are from
$\mathbb{Z} \,\cup\, \{ r \}$ for an arbitrary irrational number
$r$~\cite{BGM12}.

For constraint-based synthesis of termination arguments for various
classes of linear lasso programs, Farkas' Lemma has been extensively
used~\cite{BMS05linrank,BMS05polyrank,CSS03,HHLP13,PR04,Rybalchenko10,SSM04},
although always in its affine form.

The first complete method of ranking function synthesis for linear
lasso programs through constraint solving was due Podelski and
Rybalchenko~\cite{PR04}.  Their approach only considers lasso programs
without stem and termination arguments in form of an affine-linear
ranking function and requires only linear constraint solving.

The idea of generating affine-linear inductive invariants via Farkas'
Lemma-transformed constraints is first presented by Colón et
al.~\cite{CSS03}.  We will take the same approach when generating
inductive supporting invariants.  This method relies on non-linear
constraint solving and some of the same authors explore an
under-approximation technique for solving these~\cite{SSM04}.

Bradley, Manna and Sipma propose a similar approach for linear lasso
programs~\cite{BMS05linrank}.  They introduce affine-linear inductive
supporting invariants to handle the stem.  Their termination argument
is a lexicographic ranking function with each component corresponding to
one loop disjunct.  This not only requires non-linear constraint
solving, but also an ordering on the loop disjuncts.  The authors
extend this approach in \cite{BMS05polyrank} by the use
of \emph{template trees}.  These trees allow each lexicographical
component to have a ranking function that decreases not necessarily in
every step, but \emph{eventually}.  This bears some resemblance to
multiphase ranking functions.

Heizmann et al.\ extend the method of Podelski and
Rybalchenko~\cite{HHLP13}.  They are the first to introduce the notion
of lasso programs. Utilizing supporting invariants analogously to
Bradley et al., they synthesize affine-linear ranking functions.  Due
to their restriction to non-decreasing invariants, the generated
constraints are linear.

A collection of example-based explanations of constraint-based
verification techniques can be found in \cite{Rybalchenko10}.  This
includes the generation of ranking functions, interpolants,
invariants, resource bounds and recurrence sets.

In \cite{BG13} Ben-Amram and Genaim discuss the synthesis of linear
ranking functions for integer lasso programs without stem.  They prove
that this problem is generally co-NP-complete and continue considering
several special cases which admit a polynomial time complexity.

\section{Structure}
\label{sec-structure}

This work is divided into eight chapters.  After this introductory
chapter, in \autoref{ch-preliminaries} we recapitulate the
mathematical foundations for ordinal numbers, formal logic and linear
arithmetic including \hyperref[thm-motzkin]{Motzkin's Transposition
Theorem}.  Following this, we formally define linear lasso programs,
invariants and notions related to termination in \autoref{ch-lassos}.
This chapter concludes with a proof that termination of linear lasso
programs is undecidable.

In \autoref{ch-templates} we introduce the notion of linear ranking
function templates and formalize a way for turning synthesized
affine-linear functions into ranking functions.  We discuss the
multiphase ranking function template and three other relevant
templates and their properties.

Given a linear lasso program and a linear ranking function template,
we describe in \autoref{ch-constraints} how to build the constraints
whose solutions are the termination argument.  We analyze the
difficulty of the generated constraints---the \emph{non-linear
dimension} in \autoref{ch-nl-constraints}.  We will prove a criterion
that enables us to reduce the number of variables that occur in
non-linear operations in the constraints.  In \autoref{ch-solving} we
discuss some methods for solving the constraints and their
computational complexity.  We motivate with help of the cylindrical
algebraic decomposition of parts of the constraints that solving them
is not necessarily difficult in practice, despite the poor worst-case
time complexity of non-linear SMT solvers.  Finally, our results are
summarized in \autoref{ch-conclusion}.

This work is meant to be read in a linear fashion, each chapter
building on the results of the previous ones.  Our most important
results are stated in
chapters \ref{ch-lassos}, \ref{ch-templates}, \ref{ch-constraints} and
\ref{ch-nl-constraints}.

%% file: 2preliminaries.tex

\clearpage
\chapter{Preliminaries}
\label{ch-preliminaries}

In this chapter we introduce the required mathematical concepts from
\hyperref[sec-ordinals]{set theory}, \hyperref[sec-logic]{formal
logic} and selected results
from \hyperref[sec-linear-arithmetic]{linear programming}.  We also
dedicate \autoref{sec-motzkin} to Motzkin's Transposition Theorem.

\section{Well-orderings and Ordinal Numbers}
\label{sec-ordinals}

The definitions and results of this section are standard knowledge in
the mathematical branch of set theory~\cite{Kunen80}.

\begin{definition}[Well-ordered set]\label{def-well-ordering}
A strict linear ordering $<$ on a set $X$ is a \emph{well-ordering}
iff every non-empty subset of $X$ has a $<$-minimal element.
\end{definition}

\begin{definition}[Ordinals]\label{def-ordinal}
A set $\alpha$ is called \emph{ordinal number} or \emph{ordinal} iff
\begin{itemize}
\item $\beta \in \alpha$ implies $\beta \subset \alpha$, and
\item $\in$ (set membership) is a well-ordering on $\alpha$.
\end{itemize}
We denote the collection of all ordinals with $\mathbf{On}$.
\end{definition}

Ordinal numbers are a method of counting indefinitely.  The first
ordinal is $\emptyset$ and for every ordinal $\alpha$ the successor is
$\{ \alpha \} \cup \alpha$.  Furthermore, the union of a collection of
ordinals is again an ordinal, therefore we can take the supremum of a
collection of ordinals via set union.

Ordinals that are not successors are called \emph{limit ordinals}.
The first limit ordinal is $\omega$.  We can define addition,
multiplication and exponentiation for ordinals coinciding with these
operations on the natural numbers (however, in general addition and
multiplication are not commutative).  This yields
\begin{align*}
\omega + \omega = \omega \cdot 2, &&
\sup\{ \omega \cdot k \mid k \in \omega \} = \omega^2.
\end{align*}
We get the sequence
\begin{align*}
0, 1, 2, \ldots, \omega, \omega + 1, \ldots, \omega \cdot
2, \omega \cdot 2 + 1, \ldots, \omega^2, \omega^2 + 1, \ldots.
\end{align*}
Note that there is such a vast number of ordinals that $\mathbf{On}$
cannot be a set\footnote{If $\mathbf{On}$ was a set, it would be an
ordinal according to \autoref{def-ordinal} and hence contain itself.
This is a contradiction to the well-foundedness of set theory.}.
Nevertheless, for purposes of computer science, we are content with
the set of countable ordinals.  This justifies the usage of
$\mathbf{On}$ like a set, for example as a codomain of functions.

Every well-ordered set is isomorphic to an ordinal number.  In this
sense the ordinals are the `mothers of all well-orderings'.  This
motivates why we may consider ordinals instead of arbitrary
well-ordered sets.

\begin{lemma}[Well-orderings and ordinals]\label{lem-well-ordering}
For every well-ordered set $(X, <)$ there is a unique ordinal $\alpha$
and a bijection $f: X \to \alpha$ such that $x < y$ iff $f(x) \in
f(y)$ for every $x, y \in X$.
\end{lemma}
\begin{proof}
See the literature on set theory, e.g.~\cite{Kunen80}.
\end{proof}

\section{First-order Logic}
\label{sec-logic}

We present a short introduction to first-order logic~\cite{EFT84} and
the notation we use in this work.
Given a set $S$ containing constants, function and relation symbols,
we define \emph{terms} and \emph{formulae} of first order logic for
$S$ recursively.  Every variable and constant is an $S$-term, and so
is the application of an $n$-ary function symbol to a sequence of $n$
terms.  The application of an $n$-ary relation symbol to $n$ $S$-terms
constitute \emph{atomary} $S$-formulae (atoms).  Formulae can be
joined together using boolean connectives
$\neg, \land, \lor, \rightarrow$, and quantified using universal
($\forall$) and existential ($\exists$) quantifiers followed by the
quantified variable.  Variables that are not bound by quantifiers in a
formula $\varphi$ are called \emph{free variables of $\varphi$}.  We
use the convention that quantifiers bind weakly (until the end of the
line), and $\land$ and $\lor$ have precedence over $\rightarrow$;
negation ($\neg$) is the strongest connective.

An $S$-structure $\mathfrak{A} = (A, (Z^\mathfrak{A})_{Z \in S})$
consists of a set $A$ called the \emph{universe of $\mathfrak{A}$},
and an interpretation $Z^\mathfrak{A}$ of every symbol $Z$ in $S$.  If
an $S$-formula $\varphi$ holds in a $S$-model $\mathfrak{A}$, we
say \emph{$\mathfrak{A}$ models $\varphi$} and write
$\mathfrak{A} \models \varphi$.  An $S$-formula $\varphi$
is \emph{satisfiable} iff an $S$-structure exists that models
$\varphi$.  If $\varphi$ is modeled by all $S$-structures, we call
$\varphi$ \emph{valid} and write $\models \varphi$.  Given a set of
$S$-formulae $T$, we write $T \models \varphi$ iff
$\mathfrak{A} \models \varphi$ for every $S$-structure $\mathfrak{A}$
that models each $\psi \in T$.
If two $S$-structures $\mathfrak{A}$ and $\mathfrak{B}$ model the same
$S$-formulae, we call $\mathfrak{A}$ and
$\mathfrak{B}$ \emph{elementarily equivalent}.

In this work we entertain a special interest in \emph{linear
arithmetic} $S_\mathrm{linear} = \{ 0, 1, +, -, \leq, = \}$
and \emph{non-linear arithmetic} $S_\mathrm{non-linear} = \{ 0, 1, +,
-, \cdot, \leq, = \}$, where $0$ and $1$ are constants, $+$, $-$,
$\cdot$ are binary function symbols and $\leq$, $=$ are binary
relations.  The usual axioms concerning ordered rings apply.
Structures we consider are the rationals $\mathbb{Q}$ and the reals
$\mathbb{R}$ together with the usual interpretations of $0$, $1$, $+$,
$-$, $\cdot$, $\leq$ and $=$ as well as their elementary equivalents.
Structures elementarily equivalent to the reals are called \emph{real
closed fields}; an example of a real closed field are the real
algebraic numbers (the field of roots of rational polynomials).

Given an $S$-formulae $\varphi$, an \emph{satisfiability modulo theory
solver} (SMT solver) is a software tool that determines whether
$\varphi$ holds in an specific $S$-structure (e.g. $\mathbb{Q}$ or
$\mathbb{R}$).  If it does, the solver outputs a valuation to the free
variables of $\varphi$.  We call the input $\varphi$
the \emph{constraint} to the solution.

For notational simplicity, and if the structure $\mathfrak{A}$ is
clear from context, we identify formulae with the sets they generate.
A formula $\varphi$ containing $n$ free variables $x_1, \ldots, x_n$,
is identified with the set
\begin{align*}
\{ (x_1, \ldots, x_n) \in A^n
  \mid \mathfrak{A} \models \varphi(x_1, \ldots, x_n) \}.
\end{align*}

For later use, we state the Compactness Theorem for first order
logic~\cite{EFT84}.

\begin{theorem}[Compactness]\label{thm-compactness}
A set of formulae $T$ is satisfiable if and only if every finite
subset of $T$ is satisfiable.
\end{theorem}

\section{Linear Arithmetic}
\label{sec-linear-arithmetic}

For the remainder of this work, fix $\mathbb{K}$ to be the field of
rational numbers $\mathbb{Q}$ or any real closed field, such as the
real numbers $\mathbb{R}$.
We use the vector $x$ to denote the variables $x_1, \ldots, x_n$.  By
convention, all vectors are column vectors.  For a vector $v$,
the \emph{transpose} will be denoted as $\tr{v}$.  A function
$f: \mathbb{K}^n \rightarrow \mathbb{K}$ is
called \emph{affine-linear} (or simply \emph{affine}) iff $f(x)
= \tr{c} x + d$ for some vector $c \in \mathbb{K}^n$ and some number
$d \in \mathbb{K}$.
We call inequalities of the form $a < b$ \emph{strict inequalities}
and inequalities of the form $a \leq b$ \emph{non-strict
inequalities}.  When either comparison operator could apply to an
equation, we use the symbol $\lhd$.

Given a matrix $A \in \mathbb{K}^{m \times n}$ and a vector
$b \in \mathbb{K}^m$, the inequality $Ax \leq b$ denotes the
conjunction of the linear inequalities
\[
\bigwedge_{i=1}^m \sum_{j=1}^n a_{i,j} x_j \leq b_i
\]
where $a_{i,j}$ denotes the entry of the matrix $A$ in row $i$ and
column $j$.  If we understand these constraints as the set of vectors
$\{ x \in \mathbb{K}^n \mid Ax \leq b \}$, they form a convex subset
of $\mathbb{K}^n$ called a \emph{polyhedron}.

\section{Motzkin's Transposition Theorem}
\label{sec-motzkin}

Intuitively, \hyperref[thm-motzkin]{Motzkin's Transposition Theorem}
states that a given system of linear inequalities has no solution if
and only if a contradiction can be derived via a positive linear
combination of the equations.

\hyperref[thm-motzkin]{Motzkin's Theorem} will be used in this work to
equivalently transform universally quantified formulae into existential
ones.  Additionally, as a side effect, the number of non-linear
multiplications (multiplications of two variables) will be greatly
reduced.

\begin{theorem}[Motzkin's Transposition Theorem~\cite{Schrijver99}]
\label{thm-motzkin}
Let $A \in \mathbb{K}^{m \times n}$, $B \in \mathbb{K}^{\ell \times
n}$, $b \in \mathbb{K}^m$, and $d \in \mathbb{K}^\ell$.
\eqref{eq-motzkin1} and \eqref{eq-motzkin2} are equivalent.
\begin{align}
&\hspace{14.7mm} \forall x \in \mathbb{K}^n.\;
  \neg (Ax \leq b \;\land\; Bx < d)
\tag{M1}\label{eq-motzkin1} \\[2mm]
&\begin{aligned}
\exists \lambda \in \mathbb{K}^m \;\exists \mu \in \mathbb{K}^\ell.\;
  &\lambda \geq 0 \;\land\; \mu \geq 0 \\
  \land\; &\tr{\lambda} A + \tr{\mu} B = 0
  \;\land\;  \tr{\lambda} b + \tr{\mu} d \leq 0 \\
  \land\; &( \tr{\lambda} b < 0 \;\lor\; \mu \neq 0 )
\end{aligned}\tag{M2}\label{eq-motzkin2}
\end{align}
\end{theorem}

Note that the formula \eqref{eq-motzkin2} contains the disjunction
$(\tr{\lambda} b < 0 \;\lor\; \mu \neq 0)$.  We call the case where
$\mu = 0$ and $\tr{\lambda} b < 0$ the \emph{classical case}, because
the formula coincides with the one of the classical version of Farkas'
Lemma (\autoref{lem-farkas-classic}).  The other case will be called
the \emph{non-classical case}.  Note that the equation $\mu \neq 0$
can equivalently be written as $\sum_i \mu_i > 0$ since $\mu$ is
already constraint to non-negative entries.

The remainder of this section is dedicated to the proof
of \hyperref[thm-motzkin]{Motzkin's Theorem}.  We motivate this proof
with two versions of Farkas' Lemma which easily follow from the strong
duality theorem of linear programming~\cite{Schrijver99}: an affine
version (\autoref{lem-farkas-affine}) and a classical version
(\autoref{lem-farkas-classic}).  However, note that the literature
typically takes the opposite route and uses Farkas' Lemma to prove the
duality theorem~\cite{Schrijver99}.

\begin{theorem}[Strong duality theorem]\label{thm-strong-duality}
Let $A \in \mathbb{K}^{m \times n}$, $b \in \mathbb{K}^m$, and
$c \in \mathbb{K}^n$.  Define the linear programming problem $P
= \{ \tr{c} x \mid A x \leq b \}$ and its dual $D = \{ \tr{b}
y \mid \tr{A} y = c, y \geq 0 \}$.  If either of $P$ or $D$ is
non-empty, then $\sup P = \inf D$.
\end{theorem}
\begin{proof}
See the literature on linear programming, e.g.\ \cite{Schrijver99}.
\end{proof}

The duality theorem holds over the theory of the reals as well as the
rationals.  This is because a polyhedron defined by inequalities
involving only rational coefficients has only rational vertices.  If a
linear programming problem (or its dual respectively) has an optimal
solution, it always has a vertex as an optimal solution; hence there
is a rational optimum~\cite{Schrijver99}.

\hyperref[thm-motzkin]{Motzkin's Theorem} states that a
given system of linear inequalities has no
solution \eqref{eq-motzkin1} if and only if a contradiction can be
derived via a positive linear combination of the
equations \eqref{eq-motzkin2}.  The two cases distinguished in the
disjunction \eqref{eq-motzkin2} correspond to a contradiction derived
using only non-strict inequalities (classical case) and a
contradiction derived using at least one strict inequality
(non-classical case).

The following affine version of Farkas' Lemma is usually applied
instead of \hyperref[thm-motzkin]{Motzkin's Theorem} in the context of
lasso
programs~\cite{BMS05linrank,BMS05polyrank,CSS03,HHLP13,PR04,Rybalchenko10,SSM04}.
\hyperref[thm-motzkin]{Motzkin's Theorem} can be seen as an adaption
of \hyperref[lem-farkas-affine]{Farkas' Lemma} to allow for strict
inequalities.  Conversely, the \hyperref[lem-farkas-classic]{classic
Farkas' Lemma} is a included in \hyperref[thm-motzkin]{Motzkin's
Theorem} as the special case where $B = 0$ and $d = 0$.

\begin{lemma}[Affine Farkas' Lemma]\label{lem-farkas-affine}
Let $A \in \mathbb{K}^{m \times n}$, $b \in \mathbb{K}^m$,
$c \in \mathbb{K}^n$, and $\delta \in \mathbb{K}$ such that $Ax \leq
b$ has a solution.  Then the following two formulae are equivalent.
\begin{align*}
\forall x \in \mathbb{K}^n.&\; Ax \leq b
  \rightarrow \tr{c} x \leq \delta \\
\exists \lambda \in \mathbb{K}^m.&\; \lambda \geq 0
  \;\land\; \tr{\lambda} A = \tr{c}
  \;\land\; \tr{\lambda} b \leq \delta
\end{align*}
\end{lemma}
\begin{proof}
We reformulate \autoref{lem-farkas-affine} in terms of linear
programming.  Let $P$ and $D$ be as in \autoref{thm-strong-duality}.
\begin{center}
Let $P \neq \emptyset$. Then
  $\sup P \leq \delta$ iff $\inf D \leq \delta$.
\end{center}
If $P$ is bounded, then $D$ is feasible and by
the \hyperref[thm-strong-duality]{strong duality theorem} their
solutions are equal.  Conversely, if $\inf D \leq \delta$, then $D$ is
feasible and the strong duality theorem asserts that $\sup
P \leq \delta$.
\end{proof}

\begin{lemma}[Classic Farkas' Lemma]\label{lem-farkas-classic}
For all $A \in \mathbb{K}^{m \times n}$ and $b \in \mathbb{K}^m$ the
following two formulae are equivalent.
\begin{align*}
\forall x \in \mathbb{K}^n.&\; \neg Ax \leq b \\
\exists \lambda \in \mathbb{K}^m.&\;
  \lambda \geq 0 \;\land\; \tr{\lambda} A = 0
  \;\land\; \tr{\lambda} b < 0
\end{align*}
\end{lemma}
\begin{proof}
We proceed analogously to the proof of \autoref{lem-farkas-affine}.
Here $P$ is infeasible, consequently its dual $D$ is unbounded and
thus attains some negative value.
\end{proof}

The following Lemma is of technical nature.  We require it for the
proof of \autoref{thm-motzkin}.

\begin{lemma}[Closure of polyhedra]\label{lem-polyhedron-closure}
Let $X = \{ x \in \mathbb{K}^n \mid Ax \leq b, Bx <
d \} \neq \emptyset$.  The smallest closed set containing $X$ is $Y
= \{ x \in \mathbb{K}^n \mid Ax \leq b, Bx \leq d \}$.
\end{lemma}
\begin{proof}
$Y$ contains $X$ and is the finite intersection of closed half-spaces
and therefore closed.  We need to show that every point in
$Y \setminus X$ is the limit of a sequence of points in $X$.
Let $y \in Y \setminus X$ and since $X$ is not empty, we can pick an
$x \in X$.  For $0 < t \leq 1$,
\begin{align*}
A(tx + (1-t)y) &= tAx + (1-t)Ay \leq tb + (1-t)b = b, \\
B(tx + (1-t)y) &= tBx + (1-t)By < td + (1-t)d = d.
\end{align*}
We conclude that $tx + (1-t)y \in X$ for all $0 < t \leq 1$.  But
\begin{align*}
\lim_{t \rightarrow 0} \big( tx + (1-t)y \big) = y,
\end{align*}
therefore $y$ is in the closure of $X$.
\end{proof}

\begin{proof}[Proof of \autoref{thm-motzkin}]
Either $Ax \leq b$ is inconsistent, then \autoref{lem-farkas-classic}
states the equivalence of \eqref{eq-motzkin1} to the classical case
(first disjunct) in \eqref{eq-motzkin2}.  Otherwise write
\begin{align*}
Bx \leq d \equiv \bigwedge_{i=1}^\ell \tr{b_i} x \leq d_i.
\end{align*}
There is a subset $S \subseteq \{ \tr{b_i} x < d_i \mid 1 \leq
i \leq \ell \}$ such that $S \cup \{ Ax \leq b \}$ is satisfiable, but
$S \cup \{ Ax \leq b \} \cup \{ \tr{b_{i_0}} x < d \}$ is not for some
$i_0$.  Write $S$ as $B'x < d'$ for a submatrix $B'$ of $B$ and a
subvector $d'$ of $d$.
\eqref{eq-motzkin1} is then equivalent to
\begin{align}
\forall x.\; Ax \leq b \;\land\; B'x < d'
  \rightarrow -\tr{b_{i_0}} x \leq -d_{i_0}.
\label{eq-motzkin-proof1}
\end{align}
This can be formulated as
\begin{align*}
X := \{ x \mid Ax \leq b, B'x < d' \} \subseteq
\{ x \mid -\tr{b_{i_0}} x \leq -d_{i_0} \} =: Z.
\end{align*}
$Z$ is a closed set, hence $X$ is contained in $Z$ iff the closure of
$X$ is.  By \autoref{lem-polyhedron-closure}, the closure of $X$ is
$\{ x \mid Ax \leq b, B'x \leq d' \}$, we can therefore replace $B'x <
d'$ with $B'x \leq d'$ in \eqref{eq-motzkin-proof1}.  $Ax \leq b \land
B'x \leq d'$ is satisfiable by assumption, hence
by \autoref{lem-farkas-affine}, we get equivalently
\begin{align*}
\exists \lambda, \mu \geq 0.\; \tr{\lambda} A + \tr{\mu} B' =
-\tr{b_{i_0}} \;\land\; \tr{\lambda} b + \tr{\mu} d' \leq -d_{i_0}.
\end{align*}
This yields an assignment for the non-classical case (second disjunct)
in \eqref{eq-motzkin2}.

Conversely, a contradiction derived from the inequalities makes
$Ax \leq b \;\land\; Bx < d$ unsatisfiable.  Assume the non-classical
case of \eqref{eq-motzkin2} holds and we have an
$x^\ast \in \mathbb{K}^n$ such that $Ax^\ast \leq b$ and $Bx^\ast <
d$.  Then
\begin{align*}
\tr{\lambda} Ax^\ast \leq \tr{\lambda} b, &&
\tr{\mu} Bx^\ast < \tr{\mu} d
\end{align*}
since $\lambda$ and $\mu$ have only non-negative entries.  This yields
the following contradiction.
\begin{align*}
0 \cdot x^\ast
  = (\tr{\lambda} A + \tr{\mu} B) x^\ast
  = \tr{\lambda} Ax^\ast + \tr{\mu} Bx^\ast
  < \tr{\lambda} b + \tr{\mu} d
  \leq 0
\quad\quad\quad\qedhere
\end{align*}
\end{proof}

%% file: 3lassos.tex

\clearpage
\chapter{Lasso Programs}
\label{ch-lassos}

In \autoref{sec-lassos-def} we introduce the notion of lasso programs
and, more relevant to this work, linear lasso programs.  Invariants
and inductive invariants are presented in \autoref{sec-invariants}, as
well as the motivation to stick to the latter when building the
constraints.  Finally, in \autoref{sec-termination} we define
termination and ranking functions and conclude this chapter with a
related undecidability result.

\section{Definition}
\label{sec-lassos-def}

\begin{definition}[Lasso program~\cite{HHLP13}]\label{def-lasso}
A \emph{lasso program} $\prog = (\stemt, \loopt)$ over
the \emph{domain $\Sigma$} consists of a set of initial states
$\stemt \subseteq \Sigma$ and a binary relation
$\loopt \subseteq \Sigma \times \Sigma$.
\end{definition}

\begin{definition}[Semantics of lasso programs]\label{def-lasso-semantics}
Let $\prog = (\stemt, \loopt)$ be a lasso program over the domain
$\Sigma$.  A \emph{state} of $\prog$ is an element
$\sigma \in \Sigma$.  An \emph{execution} of $\prog$ is a (possibly
infinite) sequence of states $\sigma_0 \sigma_1 \ldots$ such that
$\sigma_0 \in \stemt$ and $(\sigma_i, \sigma_{i+1}) \in \loopt$ for
all $i \geq 0$.
\end{definition}

\begin{figure}[ht]
\centering
\begin{tikzpicture}
\node (1) at (0,0) {};
\node[state] (2) at (2.5,0) {};
\draw[arrows={-triangle 45}] (1) to node[above] {$\stemt$} (2);
\draw[arrows={-triangle 45}, loop right,in=30,out=-30,looseness=10]
  (2) to node {$\loopt$} (2);
\end{tikzpicture}
\caption{The name `lasso program' is motivated by the shape of their
  transition graph.}
\label{fig-lasso}
\end{figure}
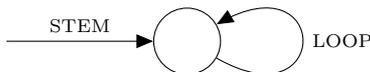

In this work, we consider the following special case of lasso
programs, namely those that have a linear specification for their stem
and loop transitions, as in the following example.

\begin{example}\label{ex-running-1}
Consider the following lasso program $\prog_{y \geq 1}$.
\begin{center}
\begin{minipage}{35mm}
\begin{lstlisting}
assume($y = 1$);
while($q \geq 0$):
    $q$ := $q - y$;
    $y$ := $y + 1$;
\end{lstlisting}
\end{minipage}
\end{center}
We can represent the stem and loop transition of $\prog_{y \geq 1}$
with the following formulae.
\begin{align*}
\stemt(q, y) &\equiv y = 1 \\
\loopt(q, y, q', y') &\equiv q \geq 0 \;\land\; q' = q - y
  \;\land\; y' = y + 1
\end{align*}
An execution of $\prog_{y \geq 1}$ is $\sigma_0 \sigma_1 \sigma_2$
where
\begin{align*}
\sigma_0&: y \mapsto 1, q \mapsto 2, \\
\sigma_1&: y \mapsto 2, q \mapsto 1, \text{ and} \\
\sigma_2&: y \mapsto 3, q \mapsto -1.
\end{align*}
Since $q$ is negative in $\sigma_2$, there is no possible successor
state to $\sigma_2$.
\end{example}

\begin{definition}[Linear lasso program]\label{def-linear-lassos}
A \emph{linear lasso program} is a lasso program $\prog =
(\stemt, \loopt)$ such that $\stemt$ and $\loopt$ are defined by
quantifier-free formulae of linear arithmetic.  A linear lasso program
is called \emph{conjunctive}, iff $\stemt$ and $\loopt$ contain no
disjunctions and negations occur only before atoms.
\end{definition}

\begin{lemma}[Linear lasso program normal form]\label{lem-lasso-nf}
For all \emph{linear lasso programs} $\prog = (\stemt, \loopt)$, the
formulae $\stemt$ and $\loopt$ can be written in the following normal
form.
\begin{align*}
\stemt(x)    &\equiv \bigvee_{n \in N}
  \big( B_n x \leq b_n
  \;\land\; B_n' x < b_n' \big) \\
\loopt(x,x') &\equiv \bigvee_{m \in M}
  \big( A_m \abovebelow{x}{x'} \leq c_m
  \;\land\; A_m' \abovebelow{x}{x'} < c_m' \big)
\end{align*}
$B_n$, $B_n'$, $A_m$, and $A_m'$ are matrices, $b_n$, $b_n'$, $c_m$,
and $c_m'$ are vectors, and $N$ and $M$ suitable finite index sets.
The program $\prog$ is conjunctive if and only if it has a normal form
with $\#N = \#M = 1$.
\end{lemma}
\begin{proof}
We transform the formulae $\stemt$ and $\loopt$ in negation normal
form such that negations occur only before atoms.  Then we rewrite
negated atoms using the following identities.
\begin{align*}
\neg a \leq b \equiv -b < -a && \neg a < b \equiv -b \leq -a && a \neq
b \equiv a < b \;\lor\; a > b
\end{align*}
Additionally, \emph{true} can be rewritten as $0 \leq 0$
and \emph{false} as $0 \leq -1$.  Finally, we transform obtained the
formulae in disjunctive normal form.
\end{proof}
According to \autoref{lem-lasso-nf}, $\stemt$ and $\loopt$ correspond
geometrically to a union of convex polyhedra
(see \autoref{fig-convex-loop}).

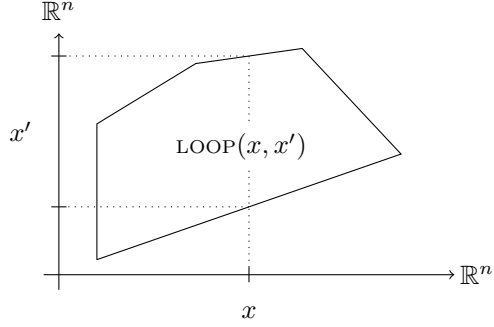
\begin{figure}[ht]
\begin{center}
\begin{tikzpicture}
\draw[->] (-0.2, 0) -- (5.2, 0);
\node at (5.5, 0) {$\mathbb{R}^n$};
\draw[->] (0, -0.2) -- (0, 3.2);
\node at (0, 3.5) {$\mathbb{R}^n$};

\draw (0.5, 0.2) -- (4.5, 1.6) -- (3.2, 3.0) -- (1.8, 2.8)
  -- (0.5, 2.0) -- (0.5, 0.2);

\draw (2.5, -0.1) -- (2.5, 0.1);
\node at (2.5, -0.5) {$x$};
\draw[dotted] (2.5, 0.1) -- (2.5, 2.9) -- (0.1, 2.9);
\draw[dotted] (2.5, 0.9) -- (0.1, 0.9);
\draw (-0.1, 2.9) -- (0.1, 2.9);
\draw (-0.1, 0.9) -- (0.1, 0.9);
\node at (-0.5, 1.9) {$x'$};

\node[fill=white] at (2.4, 1.7) {$\loopt(x, x')$};
\end{tikzpicture}
\end{center}
\caption{The loop transition of a conjunctive linear lasso program
geometrically corresponds to a polyhedron.  The $n$-dimensional state
spaces $\mathbb{R}^n$ of $x$ and $x'$ are shown compactly as either
axis.  The successor state $x'$ to a state $x$ is chosen
non-deterministically from the possible pairs $(x, x')
\in \loopt$.}
\label{fig-convex-loop}
\end{figure}

\begin{example}\label{ex-running-2}
The program $\prog_{y \geq 1}$ from \autoref{ex-running-1} is a
conjunctive linear lasso program.  Its normal form is
\begin{align*}
\stemt(q, y) \equiv\; &y \leq 1 \;\land\; -y \leq -1, \\
\loopt(q, y, q', y') \equiv\; &-q \leq 0 \;\land\; q' - q + y \leq 0
  \;\land\; -q' + q - y \leq 0 \\
  &\;\land\; y' - y - 1 \leq 0 \;\land\; -y' + y + 1 \leq 0.
\end{align*}
\end{example}

\section{Invariants}
\label{sec-invariants}

Informally, an invariant is a property that always holds during
program execution.  Although our primary goal is to prove termination,
the inference of invariants can uncover information critical to this
goal.  In \autoref{sec-augmentation} we will discuss how we involve
invariants in the ranking function discovery process.

\begin{definition}[Invariant]\label{def-invariant}
A state $\sigma \in \Sigma$ of a lasso program $\prog$
is \emph{reachable} iff there is an execution of $\prog$ containing
$\sigma$.  A formula $\psi$ is called an \emph{invariant} of $\prog$
iff $\models \psi(\sigma)$ for all reachable states $\sigma$ of
$\prog$.
\end{definition}

\begin{definition}[Affine-linear invariant]\label{def-al-invariant}
An invariant $\psi(x)$ is an \emph{affine-linear invariant} if it is
of the form
\begin{align*}
\psi(x) \equiv \tr{s}x + t \rhd 0
\end{align*}
for some vector $s \in \mathbb{K}^n$, some value $t \in \mathbb{K}$,
and $\rhd \in \{ >, \geq \}$.  If $\rhd =\; >$, the invariant $\psi(x)$
is called \emph{strict} invariant; if $\rhd =\; \geq$, the invariant
$\psi(x)$ is called \emph{non-strict} invariant.
\end{definition}

\begin{definition}[Inductive invariant]\label{def-inductive-invariant}
A formula $\psi$ is called an \emph{inductive invariant} for a linear
lasso program $\prog$ iff the following two formulae hold.
\begin{align}
\forall \sigma \in \Sigma.&\;
\stemt(\sigma) \rightarrow \psi(\sigma) \label{eq-ii}\tag{II} \\
\forall \sigma, \sigma' \in \Sigma.&\;
\psi(\sigma) \land \loopt(\sigma, \sigma') \rightarrow \psi(\sigma')
\label{eq-ic}\tag{IC}
\end{align}
\end{definition}

\begin{example}\label{ex-running-3}
$\prog_{y \geq 1}$ from \autoref{ex-running-1} has the affine-linear
inductive invariant $y \geq 1$, since it is implied by the stem and
\begin{align*}
\mathbb{K} \models \forall q, y, q', y'.\;
  y \geq 1 \;\land\; (q \geq 0 \;\land\; q' = q - y
  \;\land\; y' = y + 1) \rightarrow y' \geq 1.
\end{align*}
\end{example}

\begin{remark}\label{rem-inductive}
Every inductive invariant is an invariant.
\end{remark}
\begin{proof}
By induction using \eqref{eq-ii} and \eqref{eq-ic}.
\end{proof}

\noindent
Invariants, that are not inductive invariants are
called \emph{non-inductive} invariants.

\begin{example}\label{ex-inductive-invariant}
The converse to \autoref{rem-inductive} does not hold: non-inductive
invariants exist.  Consider the program $\prog_\mathrm{inv}$:
\begin{align*}
\stemt(y, z) &\equiv y \geq 0 \;\land\; z \geq 0 \\
\loopt(y, z, y', z') &\equiv y' = z \;\land\; z' = y
\end{align*}
$y \geq 0$ and $z \geq 0$ are invariants of $\prog_\mathrm{inv}$:
initially $y$ and $z$ are non-negative and their values do not
decrease in the loop transition.  However, neither of the two
invariants is inductive since
\begin{align*}
y \geq 0 \;\land\; y' = z \;\land\; z' = y \rightarrow y' \geq 0
\end{align*}
is false for $y = z' = 1$, and $z = y' = -1$ (and analogously for
$z \geq 0$).  Intuitively, the conclusion $y' \geq 0$ depends on the
information that $z \geq 0$ and vice versa, so neither invariant can
be proven inductively on their own.
\end{example}

\stepcounter{theorem}

\section{Termination and Ranking Functions}
\label{sec-termination}

\begin{definition}[Termination]\label{def-termination}
A lasso program \emph{terminates} iff it has no execution of infinite
length.
\end{definition}

\begin{definition}[Ranking function]\label{def-rf}
Let $\alpha$ be a set with well-ordering relation $<_\alpha$.
A \emph{ranking function} $f$ for a lasso program $\prog =
(\stemt, \loopt)$ on domain $\Sigma$ is a function
$f: \Sigma \to \alpha$ such that for all reachable states $\sigma$ and $\sigma'$
\begin{align}
\loopt(\sigma, \sigma') \rightarrow f(\sigma') <_\alpha f(\sigma).
\label{eq-rf}\tag{RF}
\end{align}
\end{definition}

If $\Sigma$ is countable, we can always make $\alpha$ countable by
choosing the image of $f$ together with the induced well-ordering on
this subset of $\alpha$.  By \autoref{lem-well-ordering} there is
always an ordinal $\beta$ and an isomorphism
$h: \alpha \rightarrow \beta$ such that $h \circ f$ is a ranking
function on $\beta$.  Without loss of generality we can therefore
assume that we are ranking over ordinals.

\begin{example}\label{ex-running-4}
The linear lasso program $\prog_{y \geq 1}$
from \autoref{ex-running-1} has the ranking function $f(q, y) = q + 1$
mapping all but the last state of every execution to a non-negative
integer.  From \autoref{ex-running-3} we know that $y \geq 1$ is an
invariant of $\prog_{y \geq 1}$, hence we can conclude that $f(q, y)$
is well-defined and decreases for every loop transition.  The ordinal
isomorphic to the non-negative integers is $\omega$, the first
infinite ordinal.
\end{example}

The next lemma illuminates the relationship between termination and
ranking functions and justifies our search for the latter for the goal
of proving termination.

\begin{lemma}\label{lem-rf-termination}
A lasso program $\prog$ has a ranking function if and only if it terminates.
\end{lemma}
\begin{proof}
The image of the states in every execution of $\prog$ under $f$ is a
strictly decreasing sequence in $\alpha$ with respect to $<_\alpha$
by \eqref{eq-rf}.  Because $<_\alpha$ is a well-ordering on $\alpha$,
this sequence cannot be infinite.

Conversely, for reachable states $\Sigma' \subseteq \Sigma$, the graph
$G = (\Sigma', \loopt)$ is acyclic by assumption.  Hence the ranking
function $f: \Sigma \to \mathbf{On}$ that assigns every state an
ordinal number such that $f(\sigma) = \sup \{ f(\sigma') \mid
(\sigma, \sigma') \in \loopt \} + 1$ is well-defined.
\end{proof}

Even though every terminating lasso program has a ranking function, in
general they can be arbitrarily complicated and their existence is
undecidable according to the following theorem.  Consequently, in this
work we want to restrict ourselves to the proper subclass of lasso
programs which are linear as well as consider only specific classes of
ranking functions.

\begin{theorem}[Halting problem for lasso programs~\cite{Tiwari04}]
\label{thm-halting}
Termination of linear lasso programs is undecidable.
\end{theorem}
\begin{proof}
We reduce the halting problem for Minsky counter
machines~\cite{Minsky61} to lasso programs.  These counter machines
have a finite number of registers (each holding one non-negative
integer) and a programming in form of a finite sequence of statements.
Possible statements are
\begin{itemize}
\item $\mathrm{INC}(r_k)$: increment register $k$ by one,
\item $\mathrm{DEC}(r_k)$: decrement register $k$ by one, and
\item $\mathrm{JZ}(r_k, s_\ell)$: if register $k$ is zero, jump to
instruction $s_\ell$, otherwise continue.
\end{itemize}
Let $M$ be such an $n$-counter machine and let its sequence of
statements be $s_0, \ldots, s_m$.  We define a linear lasso program
$\prog = (\stemt, \loopt)$ over the variables $s, r_1, \ldots, r_n$ as
follows.
\begin{align*}
\stemt \equiv s = 0 \land \bigwedge_{j=1}^n r_j = 0
\end{align*}
The loop transition $\loopt$ is a large disjunction composed of the
following disjuncts constructed from the program instructions of $M$.
\begin{align*}
s_i &= \mathrm{INC}(r_k): && \big( s = i \;\land\; s' = s + 1
  \;\land\; r_k' = r_k + 1 \;\land\; \bigwedge_{j \neq k} r_j' = r_j \big) \\
s_i &= \mathrm{DEC}(r_k): && \big( s = i \;\land\; s' = s + 1
  \;\land\; r_k \geq 1 \;\land\; r_k' = r_k - 1
  \;\land\; \bigwedge_{j \neq k} r_j' = r_j \big) \\
&~ && \lor\; \big( s = i \;\land\; s' = s + 1 \;\land\; r_k < 1
  \;\land\; r_k' = 0 \;\land\; \bigwedge_{j \neq k} r_j' = r_j \big) \\
s_i &= \mathrm{JZ}(r_k, s_\ell): && \big( s = i \;\land\; s' = \ell
  \;\land\; r_k = 0 \;\land\; \bigwedge_j r_j' = r_j \big) \\
&~ && \lor\; \big( s = i \;\land\; s' = s + 1 \;\land\; r_i \neq 0
  \;\land\; \bigwedge_j r_j' = r_j \big)
\end{align*}
Given a run for the counter machine $M$ starting with empty registers
at instruction $0$, we can construct an execution for the lasso
program $\prog$ by assigning the current program position to $s$ and
the register content to $r_1, \ldots, r_n$.  Conversely, given an
execution of $\prog$, we conclude inductively that in every state the
program counter $s$ and the registers $r_1, \ldots, r_n$ contain only
integers.  Hence we can define a run of $M$ such that every execution
step of $M$ is given by a state of $\prog$.
\end{proof}

Because of this fundamental undecidability, any method trying to prove
termination of a given lasso program must be incomplete.  However,
Braverman showed the decidability of the termination of deterministic
linear lasso programs that have an affine-linear function as loop
transition~\cite{Braverman06}, extending the work of
Tiwari~\cite{Tiwari04}.  We conjecture that the termination of
general, conjunctive linear lasso programs is also decidable.  The
critical property here seems to be the convexity of the loop
transition.  Like linear functions, polyhedral transitions tend to
move variables into a particular direction (e.g. $y' \geq y + 1$) or
rotate them about (e.g. $y' = -y$);
see \autoref{ex-multiphase-rotation}.  If one could eliminate the
rotating behavior, any terminating linear lasso program should have a
multiphase ranking function (see \autoref{sec-rft-multiphase} for its
definition): since it is terminating, there must be an inequality
$\tr{a}x + b \geq 0$ that is eventually violated.  Hence the loop
implies $\tr{a}x' \leq \tr{c}x + e$ for some $c, e$ and we can proceed
to argument about $\tr{(c - a)}x + e \geq 0$ recursively.

\begin{conjecture}\label{conj-lasso}
Termination of conjunctive linear lasso programs over rational and
real variable domain is decidable.
\end{conjecture}

Decidability of termination does hold for integer domains if the lasso
program's coefficients allow real numbers~\cite{BGM12}.

%% file: 4templates.tex

\clearpage
\chapter{Ranking Function Templates}
\label{ch-templates}

This chapter is centered around the notion of \emph{ranking function
templates}.  The concept is introduced in \autoref{sec-rft} together
with the auxiliary concept of transforming affine-linear functions to
functions with ordinals as image.  We then discuss three important
examples of linear ranking function templates, multiphase
in \autoref{sec-rft-multiphase}, piecewise in \autoref{sec-rft-pw} and
lexicographic in \autoref{sec-rft-lex}.  An overview over the results
on our ranking function templates is given
in \autoref{sec-rft-overview}.

\section{Definition}
\label{sec-rft}

\begin{definition}[Ranking function template]\label{def-rft}
A quantifier-free formula $\T(x, x')$ containing function symbols and
variables is called a \emph{ranking function template} iff for every
lasso program $\prog = (\stemt, \loopt)$, the satisfiability of
\begin{align}
\forall \sigma, \sigma' \in \Sigma.\;
  \loopt(\sigma, \sigma') \rightarrow \T(\sigma, \sigma')
\label{eq-rft}
\end{align}
implies that $\prog$ terminates.
If \eqref{eq-rft} holds for a program $\prog$, we say that
$\prog$ \emph{instantiates the template $\T$}.
\end{definition}

The ranking function template is our instrument for proving
termination.  An assignment to the function symbols and variables
gives rise to a ranking function.  Together with a set of supporting
invariants, this constitutes a termination argument.  All ranking
function templates we consider can be encoded in linear arithmetic.

We use the term \emph{affine-linear function symbol} $f(x)$ as a
shorthand for $\tr{s} x + t$ for a vector $s \in \mathbb{K}^n$ and a
variable $t \in \mathbb{K}$.

\begin{definition}[Linear ranking function template]
\label{def-linear-rft}
Let $D$ be a finite set of variables and let $F$ be a finite set of
affine-linear function symbols.  A \emph{linear ranking function
template} $\T(x, x')$ over $F$ and $D$ is a ranking function template
that can be written as a boolean combination of atoms of the form
\begin{align*}
\sum_{f \in F} \big( \alpha_f \cdot f(x) + \beta_f \cdot f(x') \big)
+ \sum_{d \in D} \gamma_d \cdot d \rhd 0,
\end{align*}
where $\alpha_f, \beta_f, \gamma_d \in \mathbb{K}$ are constants and
$\rhd \in \{ \geq, > \}$.  A variable $f \in F$ (respectively $d \in
D$) \emph{occurs in an atom $A$} of $\T(x, x')$ iff it has a non-zero
coefficient $\alpha_f$ or $\beta_f$ (respectively $\gamma_d$) in $A$.
\end{definition}

For brevity we will also write \emph{template} instead of ranking
function template and \emph{linear template} instead of linear ranking
function template.  Moreover, we disallow empty atoms of the form
$0 \rhd 0$ for notational convenience.

In order to establish that a formula conforming to the syntactic
requirements is indeed a ranking function template, \eqref{eq-rft}
must entail termination of the linear lasso program $\prog$.
According to \autoref{lem-rf-termination}, we can equally well show
that the satisfiability of \eqref{eq-rft} implies the existence of a
ranking function for $\prog$.

\begin{example}\label{ex-rft}
The formula $false$ is a ranking function template:
\begin{align*}
\forall \sigma, \sigma' \in \Sigma.\; \loopt(\sigma, \sigma')
  \rightarrow false
\end{align*}
is satisfiable iff $\loopt \equiv false$ and hence there can be no
execution of length greater than $1$ and therefore $\prog$ terminates.
$false$ is even a linear template; it can be written as $\delta >
0 \;\land\; -\delta > 0$ with variables $D = \{ \delta \}$.
\end{example}

The following linear template is applied by Podelski and Rybalchenko
in \cite{PR04}.

\begin{definition}\label{def-rft-affine}
We define the \emph{affine ranking function template} (affine
template) over the function symbols $F = \{ f \}$ and variables $D
= \{ \delta \}$ as
\begin{align}
\delta > 0 \;\land\; f(x) > 0 \;\land\; f(x') < f(x) - \delta.
\tag{$\T_\mathrm{affine}$}\label{eq-rft-affine}
\end{align}
\end{definition}

We will argue in \autoref{lem-rft-affine} that the affine template is
indeed a ranking function template; let us now check the additional
syntactic requirements for \ref{eq-rft-affine} to be a \emph{linear}
ranking function template.
\begin{align*}
\delta > 0 &\equiv
  \big( 0 \cdot f(x) + 0 \cdot f(x') \big)
  + 1 \cdot \delta > 0 \\
f(x) > 0 &\equiv
  \big( 1 \cdot f(x) + 0 \cdot f(x') \big)
  + 0 \cdot \delta > 0 \\
f(x') < f(x) - \delta &\equiv
  \big( 1 \cdot f(x) + (- 1) \cdot f(x') \big)
  + (-1) \cdot \delta > 0
\end{align*}
Thus we can write every atom of \ref{eq-rft-affine} in the required
form.

\begin{example}\label{ex-running-5}
Consider the program $\prog_{y \geq 1}$ from \autoref{ex-running-1}.
We check if $\prog_{y \geq 1}$ instantiates the affine
template \ref{eq-rft-affine}:
\begin{align*}
\forall q, y, q', y'.\; &(q \geq 0 \;\land\; q' = q - y
  \;\land\; y' = y + 1) \\
&\rightarrow (\delta > 0 \;\land\; f(q, y) > 0
  \;\land\; f(q', y') < f(q, y) - \delta)
\end{align*}
This formula is not satisfiable; essentially because it cannot be
inferred that $y$ is positive.  However, in \autoref{ex-running-3} we
showed that $y \geq 1$ is an invariant of $\prog_{y \geq 1}$ and hence
we can regard the semantically equivalent loop transition
\begin{align*}
\loopt'(q, y, q', y')
&\equiv y \geq 1 \;\land\; \loopt(q, y, q', y') \\
&\equiv y \geq 1 \;\land\; q \geq 0 \;\land\; q' = q - y
  \;\land\; y' = y + 1.
\end{align*}
Now \ref{eq-rft-affine} can be instantiated for $f(q, y) = q + 1$ and
$\delta = \frac{1}{2}$ from \autoref{ex-running-4} yielding the
following valid formula.
\begin{align*}
\forall q, y, q', y'.\;  &(y \geq 1 \;\land\; q \geq 0
  \;\land\; q' = q - y \;\land\; y' = y + 1) \\
&\rightarrow \Big( 1 > 0 \;\land\; q + 1 > 0
  \;\land\; q' + 1 < q + 1 - \frac{1}{2} \Big)
\end{align*}
\end{example}

\begin{example}\label{ex-delta}
Why do we need the positive variable $\delta$ in \ref{eq-rft-affine}?
Assume we use the following template:
\begin{align}
f(x) > 0 \;\land\; f(x') < f(x)
\label{eq-ex-delta}
\end{align}
The formula \eqref{eq-ex-delta} does not imply termination as required
by \autoref{def-rft}: $f$ could exhibit zeno behavior by attaining the
sequence of positive values
\begin{align*}
1,\; \frac{1}{2},\; \frac{1}{4},\; \frac{1}{8},\; \ldots
\end{align*}
and hence permit infinite executions.
\end{example}

Because our ranking function templates are constructed from
affine-linear function symbols, we define a conversion to functions
with the ordinal $\omega$ as image.  Transforming affine-linear
functions in this fashion yields a well ordering on their image (the
well-ordering $\in$ on ordinals).  From these `elementary' ranking
functions we will construct the ranking functions associated with the
templates.
\begin{definition}\label{def-ordinal-ranking}
Given an affine-linear function $f$ and a number $\delta > 0$ called
the \emph{step size of $f$}, we define the \emph{ordinal ranking
equivalent} of $f$ as
\begin{align}
\widehat{f}(x) =
\begin{cases}
\lceil \frac{f(x)}{\delta} \rceil, & \text{if } f(x) > 0,
\text{ and} \\
0 & \text{otherwise.}
\end{cases}
\tag{Rk}\label{eq-rank}
\end{align}
\end{definition}

$\lceil \cdot \rceil$ denotes the ceiling function that assigns to
every real number $r$ the smallest natural number that is larger or
equal to $r$.  Since the natural numbers coincide with the finite
ordinals, we can use $\lceil \cdot \rceil$ to convert a real number
into an ordinal.  Ordinal ranking equivalents are well-defined;
$\frac{f(x)}{\delta}$ is positive for $f(x) > 0$ since $\delta > 0$.
Although ordinal ranking equivalents depend on the step size, for
notational simplicity we do not explicitly denote it in $\widehat{f}$.

\begin{example}\label{ex-running-6}
Consider the ranking function $f(q, y) = q + 1$ of step size $\delta = \frac{1}{2}$
from \autoref{ex-running-5}.  Its ordinal ranking equivalent is
\begin{align*}
\widehat{f}(q, y) =
\begin{cases}
\lceil 2(q + 1) \rceil, & \text{if } q + 1 > 0, \text{ and} \\
0 & \text{otherwise.}
\end{cases}
\end{align*}
\end{example}

A formula $\T$ is a ranking function template if its satisfiability
gives rise to a ranking function.  We use ordinal ranking equivalents
to transform the assignment to the function symbols from $\T$ to
functions over ordinals.  From these we build the ranking function;
the image of this ranking function is itself an ordinal and we call
this ordinal the \emph{ranking structure of $\T$}.

\begin{lemma}\label{lem-ordinal-ranking}
Let $f$ be an affine-linear function of step size $\delta > 0$ and let
$x$ and $x'$ be two states.  If $f(x) > 0$ and $f(x) - f(x')
> \delta$, then $\widehat{f}(x) > 0$ and $\widehat{f}(x)
> \widehat{f}(x')$.
\end{lemma}
\begin{proof}
From $f(x) > 0$ follows that $\widehat{f}(x) > 0$.  Hence
$\widehat{f}(x) > \widehat{f}(x')$ in the case $\widehat{f}(x') = 0$.
For $\widehat{f}(x') > 0$, we use the fact that $f(x) - f(x')
> \delta$ to conclude that $\frac{f(x)}{\delta} - \frac{f(x')}{\delta}
> 1$ and hence $\widehat{f}(x') > \widehat{f}(x)$.
\end{proof}

We can immediately apply this lemma to show that \ref{eq-rft-affine} is
indeed a ranking function template.

\begin{lemma}\label{lem-rft-affine}
\ref{eq-rft-affine} is a linear ranking function template.
\end{lemma}
\begin{proof}
If \ref{eq-rft-affine} is implied by the loop, the assignment to $f$
and $\delta$ satisfies the requirements
of \autoref{lem-ordinal-ranking}.  Consequently, $\widehat{f}$ is a
ranking function for $\prog$ of step size $\delta$.
\end{proof}

\begin{example}\label{ex-rft-affine2}
Consider the simple non-conjunctive program $\prog_\mathrm{disj}$.
\begin{center}
\begin{minipage}{45mm}
\begin{lstlisting}
while ($q \geq 0$):
    if ($y > 0$):
        $q$ := $q - y - 1$;
    else:
        $q$ := $q + y - 1$;
\end{lstlisting}
\end{minipage}
\end{center}
Written as a linear lasso program, the stem and loop transitions are
\begin{align*}
\stemt \equiv\; \quad &\; true, \\
\loopt \equiv\; \quad &(q \geq 0 \;\land\; y > 0
    \;\land\; y' = y \;\land\; q' = q - y - 1) \\
  \lor\; &(q \geq 0 \;\land\; y \leq 0
    \;\land\; y' = y \;\land\; q' = q + y - 1).
\end{align*}
As \ref{eq-rft-affine} is implied by the loop: it has the assignment
$f(q, y) = q + 1$ and $\delta = \frac{1}{2}$.  This corresponds to the
ordinal ranking equivalent
\begin{align*}
r(q, y) = \widehat{f}(q, y) = \lceil q + 1 \rceil.
\end{align*}
\end{example}

\section{Multiphase Template}
\label{sec-rft-multiphase}

The multiphase ranking function template is targeted at programs that
go through different phases in their execution.  Each phase is ranked
with an affine-linear ranking function and the phase is considered to
be completed once this ranking function becomes non-positive.  This
yields a ranking structure of $\omega \cdot k$ as an $\omega$-ranking
is performed for each of the $k$ phases.

\begin{example}\label{ex-2phase-1}
Consider the program $\prog_{2\mathrm{-phase}}$
from \autoref{fig-multiphase-introduction}.
\begin{center}
\begin{minipage}{30mm}
\begin{lstlisting}
while ($q \geq 0$):
    $q$ := $q - y$;
    $y$ := $y + 1$;
\end{lstlisting}
\end{minipage}
\end{center}
Every execution of $\prog_{2\mathrm{-phase}}$ can be partitioned into
two phases; first $y$ increases until it is positive and then $q$
decreases until the loop condition $q \geq 0$ is violated.  Depending
on the initial values of $y$ and $q$, either phase might be skipped
altogether.
\end{example}

\begin{definition}\label{def-rft-multiphase}
We define the \emph{$k$-phase ranking function template} ($k$-phase
template) over the functions $F = \{ f_1, \ldots, f_k \}$ and
variables $D = \{ \delta_1, \ldots, \delta_k \}$ as follows.
\begin{align}
\begin{aligned}
&\bigwedge_{i=1}^k \delta_i > 0 \\
\land\; &\bigvee_{i=1}^k f_i(x) > 0 \\
\land\; &\bigwedge_{i=1}^k \Big( f_i(x') < f_i(x) - \delta_i \;\lor\;
  \bigvee_{j=1}^{i-1} f_j(x) > 0 \Big)
\end{aligned}
\tag{$\T_{k\mathrm{-phase}}$}\label{eq-rft-multiphase}
\end{align}
\end{definition}
The multiphase ranking function given by an assignment to the template
$f_1, \ldots, f_k$ to \ref{eq-rft-multiphase} is in phase $i$ if
$f_i(x) > 0$ and $f_j(x) \leq 0$ for all $j < i$.  Line 2
in \ref{eq-rft-multiphase} states that the multiphase ranking function
is always in some phase $i$.  Line 3 states that if we are in a phase
$\geq i$, then $f_i$ has to be decreasing by at least $\delta_i > 0$.

Note that the $1$-phase template coincides with the affine template.

\begin{lemma}\label{lem-rft-multiphase}
\ref{eq-rft-multiphase} is a linear ranking function template.
\end{lemma}
\begin{proof}
It is clear that \ref{eq-rft-multiphase} conforms to the syntactic
requirements to be a linear template.  Consider the following ranking
function on $\omega \cdot k$.
\begin{align}
r(x) =
\begin{cases}
\omega \cdot (k - i) + \widehat{f_i}(x) & \text{if }
f_j(x) \leq 0 \text{ for all } j < i \text{ and } f_i(x) > 0, \\
0 & \text{otherwise.}
\end{cases}
\label{eq-rft-multiphase-rf}
\end{align}
Let $(x, x') \in \loopt$.  We need to show that $r(x') < r(x)$.  From
line 2 in $\T_{k\mathrm{-phase}}$ follows that $r(x) > 0$ for any
$x$, and there is an $i$ such that $f_i(x) > 0$ and $f_j(x) \leq 0$
for all $j < i$.  By line 3, $f_j(x') \leq 0$ for all $j < i$ because
$f_j(x') < f_j(x) - \delta_j \leq 0 - \delta_j \leq 0$ since
$f_\ell(x) \leq 0$ for all $\ell < j$.

If $f_i(x') \leq 0$, then $r(x') \leq \omega \cdot (k - i)
< \omega \cdot (k - i) + \widehat{f_i}(x) = r(x)$.  Otherwise $f_i(x')
> 0$ and from line 3 follows $f_i(x') < f_i(x) - \delta_i$.
By \autoref{lem-ordinal-ranking}, $\widehat{f_i}(x)
> \widehat{f_i}(x')$ for the ordinal ranking equivalent $f_i$ with
step size $\delta_i$.  Hence
\begin{align*}
r(x') = \omega \cdot (k - i) + \widehat{f_i}(x')
  < \omega \cdot (k - i) + \widehat{f_i}(x)
  = r(x). \qquad \qedhere
\end{align*}
\end{proof}

\begin{example}\label{ex-2phase-2}
Consider the program $\prog_{2\mathrm{-phase}}$
from \autoref{ex-2phase-1}.  From the $2$-phase template we get an
assignment $f_1(q, y) \mapsto 1 - y$ and $f_2(q, y) \mapsto q + 1$,
each with step size $1$.  Thus $\prog_{2\mathrm{-phase}}$ has the
ranking function
\begin{align*}
r(q, y) =
\begin{cases}
\omega + \lceil 1 - y \rceil, & \text{if } y < 1, \\
\lceil q + 1 \rceil, & \text{if } y \geq 1 \;\land\; q + 1 > 0,
  \text{ and} \\
0 & \text{otherwise.}
\end{cases}
\end{align*}
\end{example}

\begin{example}\label{ex-multiphase-rotation}
There are terminating conjunctive linear lassos that do not have a
multi-phase ranking function:
\begin{center}
\begin{minipage}{43mm}
\begin{lstlisting}
assume($z \geq y + 1$);
while($q \geq 0$):
    $q$ := $q + z - y - 1$;
    $y$ := $-y$;
    $z$ := $-z$;
\end{lstlisting}
\end{minipage}
\end{center}
Here $y$ and $z$ are both subject to a rotation of $180$ degrees.  The
function $f(q, y, z) = q + 1$ is eventually decreasing in the sense
that after a finite number of iterations, its value will have
decreased.  However, during one step its value might increase.  If we
consider the loop transition $\loopt' = \loopt \circ \loopt$ such that
the loop body is executed twice, then $f$ is indeed a ranking function
since $y$ and $z$ remain constant.  However, concatenating the loop
does not work in general since a variables $y$ and $z$ can do a
rotation by an arbitrary irrational angle $\alpha$ (even over the
theory of the rationals):
\begin{align*}
y' = \cos(\alpha) \cdot y - \sin(\alpha) \cdot z
\;\land\; z' = \sin(\alpha) \cdot y + \cos(\alpha) \cdot y
\end{align*}
Consequently, there is not necessarily a finite number of
concatenations of $\loopt$ that make $y$ remain constant.
\end{example}

\begin{example}\label{ex-multiphase-complexity}
Although every phase has a linear ranking function, we cannot use this
to state a complexity result about the program in question.  The
reason is the non-determinism of linear lasso programs. Consider the
following linear lasso program.
\begin{align}
\begin{aligned}
\stemt(q, y) \equiv\; &y = 1 \\
\loopt(q, y, q', y') \equiv\;
  &(q \geq 0 \;\land\; y \geq 1 \;\land\; y' = 0) \;\lor\; \\
  &(q \geq 0 \;\land\; y \leq 0 \;\land\; y' = y - 1 \;\land\; q' = q - 1)
\end{aligned}
\tag{$\prog_\mathrm{runtime}$}\label{eq-prog-runtime}
\end{align}
The runtime of \ref{eq-prog-runtime} does not depend on the input at
all: after the first loop execution $y$ is set to $0$ and $q$ is set
to \emph{some arbitrary value}.  In particular, this value does not
depend on the initial value of $q$.  The remainder of the loop
execution then takes $\lceil q \rceil + 1$ iterations to terminate.

However, \ref{eq-prog-runtime} instantiates the $2$-phase template: it
has the $2$-phase ranking function $f_1(q, y) = y$ and $f_2(q, y) = q
+ 1$.  It provably terminates, there is just no a priori bound on the
execution steps.
\end{example}

\section{Piecewise Template}
\label{sec-rft-pw}

The piecewise ranking function template formalizes an affine-linear
ranking function that is defined piecewise using affine-linear
predicates to discriminate the different pieces.  This discrimination
need not be unambiguous; if two predicates overlap, their
corresponding affine-linear functions are both ranking functions.
Piecewise ranking functions have a ranking structure of $\omega$.

\begin{definition}\label{def-rft-pw}
We define the \emph{$k$-piece ranking function template} ($k$-piece
template) over the functions $F = \{ f_1, \ldots, f_k, g_1, \ldots,
g_k \}$ and variables $D = \{ \delta \}$ as follows.
\begin{align}
\begin{aligned}
&\delta > 0 \\
\land\; &\bigwedge_{i=1}^k \bigwedge_{j=1}^k \Big( g_i(x) < 0
  \;\lor\; g_j(x') < 0 \;\lor\; f_j(x') < f_i(x) - \delta \Big) \\
\land\; &\bigwedge_{i=1}^k f_i(x) > 0 \\
\land\; &\bigvee_{i=1}^k g_i(x) \geq 0
\end{aligned}
\tag{$\T_{k\mathrm{-piece}}$}\label{eq-rft-pw}
\end{align}
\end{definition}
We call the function symbols $\{ g_i \mid 1 \leq i \leq
k \}$ \emph{discriminating predicates} and the function symbols $\{
f_i \mid 1 \leq i \leq k \}$ \emph{ranking pieces}.

Line 4 of \ref{eq-rft-pw} states that the predicates cover all states;
in other words, the piecewise defined ranking function is not just a
partial function.  Given the the $k$ different pieces $f_1, \ldots,
f_k$ and a state $x$, we use $f_i$ as a ranking function only if
$g_i(x) \geq 0$.  This choice need not be unambiguous---the
discriminating predicates may overlap.  If they do, we can use any one
of their ranking pieces.  According to line 3 in \ref{eq-rft-pw}, all
ranking pieces are positive-valued and by line 2 piece transitions are
well-defined: the rank of the new state is always less than the rank
any of the ranking pieces assigned to the old state.  We formally
prove this in the following lemma.

\begin{lemma}\label{lem-rft-pw}
\ref{eq-rft-pw} is a linear ranking function template.
\end{lemma}
\begin{proof}
It is clear that \ref{eq-rft-pw} conforms to the syntactic
requirements to be a linear template.  Consider the following ranking
function on $\omega$.
\begin{align}
r(x) = \max \{ \widehat{f_i}(x) \mid g_i(x) \geq 0 \}
\end{align}
The function $r$ is well-defined because according to line 4
in \eqref{eq-rft-pw}, the set $\{ \widehat{f_i}(x) \mid g_i(x) \geq
0 \}$ is not empty.  Let $(x, x') \in \loopt$ and let $i$ and $j$ be
indices such that $r(x) = \widehat{f_i}(x)$ and $r(x')
= \widehat{f_j}(x')$.  By definition of $r$, we have that $g_i(x) \geq
0$ and $g_j(x) \geq 0$ and line 2 then implies $f_j(x') < f_i(x)
- \delta$.  According to \autoref{lem-ordinal-ranking} and line 3,
this entails $\widehat{f_j}(x') < \widehat{f_i}(x)$ and thus $r(x') <
r(x)$.
\end{proof}

\begin{example}\label{ex-rft-pw}
Consider the following program $\prog_\mathrm{gcd}$ adapted
from \cite{BMS05linrank}.
\begin{center}
\begin{minipage}{56mm}
\begin{lstlisting}
assume($y_1 \geq 1 \;\land\; y_2 \geq 1$);
while($y_1 - y_2 \geq 1 \lor y_2 - y_1 \geq 1$):
    if ($y_1 > y_2$):
        $y_1$ := $y_1 - y_2$;
    else:
        $y_2$ := $y_2 - y_1$;
\end{lstlisting}
\end{minipage}
\end{center}
Given two positive integers $y_1$ and $y_2$, the program
$\prog_\mathrm{gcd}$ computes the greatest common denominator.  Note
that $y_1 - y_2 \geq 1 \lor y_2 - y_1 \geq 1$ is the integer
equivalent of $y_1 \neq y_2$.

The program $\prog_\mathrm{gcd}$ instantiates the 2-piece template for
the ranking pieces $f_1(y_1, y_2) = y_1$ and $f_2(y_1, y_2) = y_2$
with step size $\delta = 1$ and discriminating predicates $g_1(y_1,
y_2) = y_1 - y_2$ and $g_2(y_1, y_2) = y_2 - y_1$, given the two
inductive invariants $y_1 \geq 1$ and $y_2 \geq 1$.
\end{example}

\section{Lexicographic Template}
\label{sec-rft-lex}

Lexicographic ranking functions are used frequently and have been
adopted to lasso programs by Bradley, Manna and
Sipma~\cite{BMS05linrank}.  They consist of lexicographically ordered
components of affine-linear functions.  Hence they have a ranking
structure of $\omega^k$.  A state is mapped to a tuple of values such
that the loop transition leads to a decrease with respect to the
lexicographic ordering for this tuple.  Therefore no function may
increase unless a function of a lower index decreases.  Additionally,
at every step there must be at least one function that decreases.

\begin{definition}\label{def-rft-lex}
We define the \emph{$k$-lexicographic ranking function template}
($k$-lexicographic template) over the functions $F = \{ f_1, \ldots,
f_k \}$ and variables $D = \{ \delta_1, \ldots, \delta_k \}$ as
follows.

\begin{align}
\begin{aligned}
&\bigwedge_{i=1}^{k} \delta_i > 0 \\
\land\; &\bigwedge_{i=1}^k f_i(x) > 0 \\
\land\; &\bigwedge_{i=1}^{k-1} \Big( f_i(x') \leq f_i(x)
  \;\lor\; \bigvee_{j=1}^{i-1} f_j(x') < f_j(x) - \delta_j \Big) \\
\land\; &\bigvee_{i=1}^k f_i(x') < f_i(x) - \delta_i
\end{aligned}
\tag{$\T_{k\mathrm{-lex}}$}\label{eq-rft-lex}
\end{align}
\end{definition}
Consider the formula \ref{eq-rft-lex}.  Line 2 establishes that all
lexicographic entries $f_1, \ldots, f_k$ are positive-valued.  In
every step, at least one component must decrease according to line 4.  All
functions corresponding to indexes smaller than the decreasing
function may increase by line 3.

\begin{example}\label{ex-rft-lex}
Consider the program $\prog_\mathrm{gcd}$ from \autoref{ex-rft-pw}.
$\prog_\mathrm{gcd}$ has the lexicographic ranking function with first
index $f_1(y_1, y_2) = y_2$ and second index $f_2(y_1, y_2) = y_1$
provided the two inductive invariants $y_1 \geq 1$ and $y_2 \geq 1$.
\end{example}

\begin{lemma}\label{lem-rft-lex}
\ref{eq-rft-lex} is a linear ranking function template.
\end{lemma}
\begin{proof}
It is clear that \ref{eq-rft-lex} conforms to the syntactic
requirements to be a linear template.  Consider the following ranking
function on $\omega^k$.
\begin{align}
r(x) = \sum_{j=1}^k \omega^{k-j} \cdot \widehat{f_j}(x)
\end{align}
Let $(x, x') \in \loopt$.  From line 2 in $\T_{k\mathrm{-lex}}$
follows $f_j(x) > 0$ for all $j$, so $r(x) > 0$.  By line 4
and \autoref{lem-ordinal-ranking}, there is a minimal $i$ such that
$\widehat{f_i}(x') < \widehat{f_i}(x)$.  Line 3 implies that
$\widehat{f_1}(x') \leq \widehat{f_1}(x)$ and hence inductively
$\widehat{f_j}(x') \leq \widehat{f_j}(x)$ for all $j < i$ since $i$
was minimal.
\begin{align*}
r(x') &= \sum_{j=1}^k \omega^{k-j} \cdot \widehat{f_j}(x') \\
  &\leq \sum_{j=1}^{i-1} \omega^{k-j} \cdot \widehat{f_j}(x)
    + \sum_{j=i}^k \omega^{k-j} \cdot \widehat{f_j}(x') \\
  &< \sum_{j=1}^{i-1} \omega^{k-j} \cdot \widehat{f_j}(x)
    + \omega^{k-i} \cdot \widehat{f_i}(x) \\
  &\leq r(x) \qedhere
\end{align*}
\end{proof}

%% file: 5constraints.tex

\clearpage
\chapter{Building the Constraints}
\label{ch-constraints}

In this chapter we discuss an automatic procedure for the
instantiation of ranking function templates.  Given a linear lasso
program $\prog$ and a linear ranking function template $\T$, we set up
constraints whose solution is a termination argument for $\prog$.
With the help of \hyperref[thm-motzkin]{Motzkin's Transposition
Theorem}, this will be a purely existentially quantified formula.  We
first discuss the addition of supporting invariants
in \autoref{sec-augmentation}.  The step-by-step transformations
involved in building the constraints are the subject
in \autoref{sec-constraints}.  In \autoref{sec-soundness-completeness}
we show that this procedure is sound and complete.  We conclude this
chapter with a discussion of lasso programs that contain integer
variables in \autoref{sec-integers}.

\section{Loop Augmentation with Invariants}
\label{sec-augmentation}

Ranking function templates are checked for implication by the loop
transition.  As this transition is independent of the lasso program's
stem, information critical to the program's termination proof might be
missed.  We address this issue by augmenting the loop transition with
inductive invariants similar to \cite{CSS03}.  The following lemma
formalizes this process.

\begin{lemma}[Loop augmentation with invariants]\label{lem-invariant-aug}
Let $\T$ be a ranking function template, $\prog = (\stemt, \loopt)$ be
a lasso program, and $(\psi_\ell)_{\ell \in L}$ a finite number of
invariants of $\prog$.  If the formula
\begin{align}
\forall x, x'.\; \big( \loopt(x, x')
    \land \bigwedge_{\ell \in L} \psi_\ell(x) \big)
  \rightarrow \T(x, x') \label{eq-rft'}
\end{align}
is satisfiable, then $\prog$ terminates.
\end{lemma}
\begin{proof}
Because every $\psi_\ell$ holds at all reachable states of $\prog$, so
does $\bigwedge_{\ell \in L} \psi_\ell$.  Consider the transition
\begin{align*}
\loopt'(x, x') \equiv \Big( \bigwedge_{\ell \in L} \psi_\ell(x) \Big)
  \;\land\; \loopt(x, x').
\end{align*}
The two programs $\prog = (\stemt, \loopt)$ and $\prog' =
(\stemt, \loopt')$ are semantically equivalent: they have the same
executions.  Therefore $\prog$ terminates iff $\prog'$ terminates, and
the latter is equivalent to the satisfiability of \eqref{eq-rft'}
by \autoref{def-rft}.
\end{proof}

In addition to solving the constraints \eqref{eq-rft'}, we need to
encode that all $\psi_\ell$ are invariants.
We use inductive invariants and add the
conditions \eqref{eq-ii} and \eqref{eq-ic} to our constraint system
for every inductive invariant $\psi_\ell$.
\begin{align}
\bigwedge_{\ell \in L} &\forall x.\; \stemt(x) \rightarrow \psi_\ell(x)
\tag{II0}\label{eq-ii0} \\
\bigwedge_{\ell \in L} &\forall x, x'.\; \psi_\ell(x)
  \;\land\; \loopt(x, x') \rightarrow \psi_\ell(x')
\tag{IC0}\label{eq-ic0} \\
&\forall x, x'.\; \big( \loopt(x, x')
    \;\land\; \bigwedge_{\ell \in L} \psi_\ell(x) \big)
  \rightarrow \T(x, x') \tag{TI0}\label{eq-ti0}
\end{align}
We call \eqref{eq-ii0} \emph{invariant
initiation}, \eqref{eq-ic0} \emph{invariant consecution}
and \eqref{eq-ti0} the \emph{template implication}.

The following lemma asserts that a finite number of inductive
invariants is sufficient to entail the ranking function template, if
it is entailed by all inductive invariants.

\begin{lemma}\label{lem-finite-invariants}
Let $I$ be the set of all inductive invariants of $\loopt$ and
$\varphi$ be a formula.  Then $I, \loopt \models \varphi$
iff there is a finite subset $I' \subseteq I$ such that
$I', \loopt \models \varphi$.
\end{lemma}
\begin{proof}
If $I, \loopt \models \varphi$, then the set $T :=
I \cup \{ \loopt, \neg\varphi \}$ is unsatisfiable. By
the \hyperref[thm-compactness]{Compactness Theorem}, there is a finite
subset $T' \subseteq T$ that is unsatisfiable, and hence $T'' =
T' \cup \{ \loopt, \neg\varphi \}$ is also unsatisfiable.  Therefore
we can conclude for the finite set $I' = T'' \cap I$ that
$I', \loopt \models \varphi$.
\end{proof}

\section{The Constraints}
\label{sec-constraints}

In this section we will sequentially apply five equivalence
transformations to the constraints \eqref{eq-ii0}, \eqref{eq-ic0}
and \eqref{eq-ti0} to make them more easily solvable by an SMT solver.
The reason for this is that the result (1) has only existential
quantification instead of universal and (2) has a significantly
reduced number of non-linear operations (multiplications of
variables).  In \autoref{sec-soundness-completeness} we argue that
each transformation is indeed an equivalence transformation; this
method is sound and complete.

We fix a linear ranking function template over $F$ in conjunctive
normal form,
\begin{align}
\T(x, x') \equiv \bigwedge_{i \in I} \bigvee_{j \in J_i} \T_{i,j}(x, x')
  \equiv \bigwedge_{i \in I} \bigvee_{j \in J_i}
    \tr{d_{i,j}} \abovebelow{x}{x'} \rhd_{i,j} e_{i,j},
\label{eq-rft-explicit}
\end{align}
where the vectors $d$ and the numbers $e$ are linear combinations of
the uninterpreted function symbols in $F$ and $\rhd_{i,j} \in \{ \geq,
> \}$.  We partition every $J_i$ in $J_i^\geq$ and $J_i^>$ such that
$\rhd_{i,j} = \geq$ for all $j \in J_i^\geq$ and $\rhd_{i,j} = >$ for
all $j \in J_i^>$.

Furthermore, we fix a linear lasso program $\prog = (\stemt, \loopt)$.
According to \autoref{lem-lasso-nf}, we can write $\prog$ in normal
form:
\begin{align}
&\hspace{5mm} \stemt(x)
  \equiv \bigvee_{n \in N} \stemt_n(x)
  \equiv \bigvee_{n \in N} \big( B_n x \leq b_n
    \;\land\; B_n' x < b_n' \big)
\label{eq-stem-explicit} \\
&\begin{aligned}
\loopt(x, x') &\equiv \bigvee_{m \in M} \loopt_m(x, x') \\
&\equiv \bigvee_{m \in M} \big( A_m \abovebelow{x}{x'} \leq c_m
  \;\land\; A_m' \abovebelow{x}{x'} < c_m' \big)
\end{aligned}
\label{eq-loop-explicit}
\end{align}
Let $\{ \psi_\ell \mid \ell \in L \}$ denote the inductive invariants;
every invariant is an inequality of the form $\tr{s} x + t \rhd 0$ for
a vector $s$, a number $t$, and $\rhd \in \{ \geq, > \}$.

\transformationStep{Remove disjunctions in stem and loop.}
The disjunction in the stem and loop formulae are moved outside the
quantifiers' scope in \eqref{eq-ii0}, \eqref{eq-ic0}
and \eqref{eq-ti0}.
\begin{align}
\bigwedge_{\ell \in L} \bigwedge_{n \in N} &\forall x.\;
  \stemt_n(x) \rightarrow \psi_\ell(x)
\label{eq-ii1}\tag{II1} \\
\bigwedge_{\ell \in L} \bigwedge_{m \in M} &\forall x, x'.\;
  \psi_\ell(x) \land \loopt_m(x, x') \rightarrow \psi_\ell(x')
\label{eq-ic1}\tag{IC1} \\
\bigwedge_{m \in M} &\forall x, x'.\;
  \loopt_m(x, x') \land \bigwedge_{\ell \in L} \psi_\ell(x)
  \rightarrow \T(x, x')
\label{eq-ti1}\tag{TI1}
\end{align}

\transformationStep{Remove template conjunctions.}
The conjunctions in the ranking function template are moved outside
the quantifiers' scope.
\begin{align}
\bigwedge_{\ell \in L} \bigwedge_{n \in N} &\forall x.\;
  \stemt_n(x) \rightarrow \psi_\ell(x)
\label{eq-ii2}\tag{II2} \\
\bigwedge_{\ell \in L} \bigwedge_{m \in M} &\forall x, x'.\;
  \psi_\ell(x) \land \loopt_m(x, x') \rightarrow \psi_\ell(x')
\label{eq-ic2}\tag{IC2} \\
\bigwedge_{i \in I} \bigwedge_{m \in M} &\forall x, x'.\;
  \loopt_m(x, x') \land \bigwedge_{\ell \in L} \psi_\ell(x)
  \rightarrow \bigvee_{j \in J_i} \T_{i,j}(x, x')
\label{eq-ti2}\tag{TI2}
\end{align}

\transformationStep{Replicate supporting invariants.}
We supply different supporting invariants to every template
implication.  This will later enable us to get rid of a number of
non-linear variables.  In order to achieve this, we introduce
invariants for every $\ell \in L$, $i \in I$ and $m \in M$; therefore
let $L' = L \times I \times M$.
\begin{align}
\bigwedge_{\ell \in L'} \bigwedge_{n \in N} &\forall x.\;
  \stemt_n(x) \rightarrow \psi_\ell(x)
\label{eq-ii3}\tag{II3} \\
\bigwedge_{\ell \in L'} \bigwedge_{m \in M} &\forall x, x'.\;
  \psi_\ell(x) \land \loopt_m(x, x') \rightarrow \psi_\ell(x')
\label{eq-ic3}\tag{IC3} \\
\bigwedge_{i \in I} \bigwedge_{m \in M} &\forall x, x'.\;
  \loopt_m(x, x') \land \bigwedge_{\ell \in L} \psi_{(\ell,i,m)}(x)
  \rightarrow \bigvee_{j \in J_i} \T_{i,j}(x, x')
\label{eq-ti3}\tag{TI3}
\end{align}

\transformationStep{Write as negated conjunctions.}
In order to make \hyperref[thm-motzkin]{Motzkin's Theorem} applicable, we
write the implications equivalently as negated conjunctions.
\begin{align}
\bigwedge_{\ell \in L'} \bigwedge_{n \in N} &\forall x.\;
  \neg (\stemt_n(x) \;\land\; \neg\psi_\ell(x))
\label{eq-ii4}\tag{II4} \\
\bigwedge_{\ell \in L'} \bigwedge_{m \in M} &\forall x, x'.\;
  \neg (\psi_\ell(x) \;\land\; \loopt_m(x, x')
    \;\land\; \neg\psi_\ell(x'))
\label{eq-ic4}\tag{IC4} \\
\bigwedge_{i \in I} \bigwedge_{m \in M} &\forall x, x'.\;
  \neg \Big(
  \loopt_m(x, x') \land \bigwedge_{\ell \in L} \psi_{(\ell,i,m)}(x)
    \land \bigwedge_{j \in J_i} \neg\,\T_{i,j}(x, x') \Big)
\label{eq-ti4}\tag{TI4}
\end{align}

\transformationStep{Apply Motzkin's Transposition Theorem.}
For simplicity we assume that no involved invariants are non-strict
inequalities (strict inequalities are processed analogously).  We
rewrite the invariants
\begin{align*}
\psi_{\ell,i,m}(x) \equiv \tr{s_{\ell,i,m}} x + t_{\ell,i,m} \geq 0
\end{align*}
where $s_{\ell,i,m} \in \mathbb{K}^n$ and
$t_{\ell,i,m} \in \mathbb{K}$ are variables.  Similarly, we
use \eqref{eq-rft-explicit}, \eqref{eq-stem-explicit},
and \eqref{eq-loop-explicit} to rewrite $\T_{i,j}$, $\stemt_n$ and
$\loopt_m$ as linear inequalities.  Next, we
apply \hyperref[thm-motzkin]{Motzkin's Transposition Theorem} to every
universally quantified subformula and obtain the following equivalent
constraints.
\begin{align}
\begin{aligned}
\bigwedge_{\ell \in L'} \bigwedge_{n \in N}
&\exists \lambda, \mu, \xi \geq 0. \\
&\;\quad \tr{\lambda} B_n + \tr{\mu} B_n'
  + \xi \tr{\abovebelow{s_\ell}{0}} = 0 \\
&\land\; \tr{\lambda} b_n + \tr{\mu} b_n' + \xi t_\ell \leq 0 \\
&\land\; \big( \tr{\lambda} b_n < 0
  \;\lor\; \xi + \sum \mu > 0 \big)
\end{aligned}\label{eq-ii5}\tag{II5}
\end{align}
\begin{align}
\begin{aligned}
\bigwedge_{\ell \in L'} \bigwedge_{m \in M}
&\exists \lambda, \chi_1, \mu, \chi_2 \geq 0.\; \\
&\;\quad \tr{\lambda} A_m + \tr{\mu} A_m'
  + \chi_2 \tr{\abovebelow{0}{s_\ell}}
  - \chi_1 \tr{\abovebelow{s_\ell}{0}} = 0 \\
&\land\; \tr{\lambda} c_m + \tr{\mu} c_m'
    + (\chi_2 - \chi_1) t_\ell \leq 0 \\
&\land\; \big( \tr{\lambda} c_m - \chi_1 t_\ell < 0
  \;\lor\; \chi_2 + \sum \mu > 0 \big)
\end{aligned}\label{eq-ic5}\tag{IC5}
\end{align}
\begin{align}
\begin{aligned}
\bigwedge_{i \in I} \bigwedge_{m \in M}
&\exists \lambda, (\xi_\ell)_{\ell \in L},
  (\zeta_j)_{j \in J_i}, \mu \geq 0. \\
&\quad\; \tr{\lambda} A_m + \tr{\mu} A_m'
  + \sum_{\ell \in L} \xi_\ell \tr{\abovebelow{s_{\ell,i,m}}{0}}
  + \sum_{j \in J_i} \zeta_j \tr{d_{i,j}} = 0 \\
&\land\; \tr{\lambda} c_m + \tr{\mu} c_m'
  + \sum_{\ell \in L} \xi_\ell t_{\ell,i,m}
  + \sum_{j \in J_i} \zeta_j e_{i,j} \leq 0 \\
&\land\; \big( \tr{\lambda} c_m
  + \sum_{\ell \in L} \xi_\ell t_{\ell,i,m}
  + \sum_{j \in J_i^\geq} \zeta_j e_{i,j} < 0 \\
&\quad\quad \lor\; \sum_{j \in J_i^>} \zeta_j + \sum \mu > 0 \big)
\end{aligned}\label{eq-ti5}\tag{TI5}
\end{align}
An explanation to the coefficients introduced
by \hyperref[thm-motzkin]{Motzkin's Transposition Theorem} is in
order.  For every inequality in \eqref{eq-motzkin1}, a new
existentially quantified variable is added in \eqref{eq-motzkin2}.  We
call these new existentially quantified variables \emph{Motzkin
coefficients}.

In the invariant initiation \eqref{eq-ii5}, the stem's non-strict
inequalities correspond to the vector of variables $\lambda$, the
stem's strict inequalities correspond to the vector of variables
$\mu$.  The invariant has the Motzkin coefficient $\xi$.
In the invariant consecution \eqref{eq-ic5} and the template
implication \eqref{eq-ti5}, the loop's non-strict inequalities
correspond to the vector of variables $\lambda$ and the loop's strict
inequalities to the vector of variables $\mu$.
In \eqref{eq-ic5} the Motzkin coefficient $\chi_1$ corresponds to the
premise $\psi_\ell(x)$, while the Motzkin coefficient $\chi_2$
corresponds to $\psi_\ell(x')$.  Lastly, the Motzkin coefficients
$\xi_\ell$ in \eqref{eq-ti5} correspond to the invariants
$\psi_\ell(x)$ in the template implication and the Motzkin
coefficients $\zeta_j$ correspond to the ranking function template's
inequalities.

\section{Soundness and Completeness}
\label{sec-soundness-completeness}

It is clear that step 1, 2 and 4 are equivalence transformations; they
are simple syntactic modifications that preserve semantics.  Step 3
introduces a number of new invariants, hence the constraints
potentially gain new solutions, but retain all old solutions.
However, adding more invariants is still sound;
by \autoref{lem-invariant-aug} we might just as well have started out
with the larger number of invariants.  Finally, transformation 5 also
retains equivalence by \hyperref[thm-motzkin]{Motzkin's Transposition
Theorem}.

We can now state the soundness and completeness of our method: solving
the constraint \eqref{eq-ii5} $\land$ \eqref{eq-ic5}
$\land$ \eqref{eq-ti5} is equivalent to solving the
constraint \eqref{eq-ii0} $\land$ \eqref{eq-ic0}
$\land$ \eqref{eq-ti0}, which according to \autoref{lem-invariant-aug}
is satisfiable only if $\prog$ terminates.

\begin{theorem}[Soundness]\label{thm-soundness}
If the constraint \eqref{eq-ii5} $\land$ \eqref{eq-ic5}
$\land$ \eqref{eq-ti5} is satisfiable, then $\prog$ terminates.
\end{theorem}

\begin{theorem}[Completeness]\label{thm-completeness}
If the constraint \eqref{eq-ii0} $\land$ \eqref{eq-ic0}
$\land$ \eqref{eq-ti0} is satisfiable, then so is the
constraint \eqref{eq-ii5} $\land$ \eqref{eq-ic5}
$\land$ \eqref{eq-ti5}.
\end{theorem}

We get even more than just termination guarantee
the \hyperref[thm-soundness]{soundness theorem} suggests.  The
resulting variable assignment gives rise to a termination argument in
form of a ranking function together with a set of supporting
invariants.  This serves as a termination proof that can be verified
by an independent theorem prover.

\section{Integer Lasso Programs}
\label{sec-integers}

In \autoref{sec-constraints} we built the constraints for rational or
real variable domains.  In this section we want to motivate that the
same procedure can be applied to integer or mixed integer variable
domains.

\begin{definition}\label{def-mixed-variable}
A lasso program is said to have \emph{mixed integer domain}, iff it
contains some variables whose domain is the integers.
\end{definition}

The soundness of this method for integers is trivial---the integers
are a subset of the rationals and hence every execution of a program
of mixed integer domain is also an execution of the program with the
larger domain.  Therefore the mixed integer program has no infinite
execution if the program with larger domain has none.  We are
interested in the completeness.  Note that even the instantiation of
the affine template for lasso programs without stem is co-NP-complete
in the integer case~\cite{BG13}.

We introduce the notion of \emph{integral polyhedra}.  A polyhedron is
integral, if it contains all inequalities that do not follow over the
rationals, but are entailed over the integers.  This will enable the
use of \hyperref[thm-motzkin]{Motzkin's Theorem} for integer
polyhedra.  We give an equivalent definition.
\begin{definition}\label{def-integral}
A polyhedron $Ax \leq b$ is \emph{integral} iff it coincides with the
convex hull of the integer solutions of $Ax \leq b$.
\end{definition}

For a given polyhedron, we can compute its \emph{integral hull} (the
corresponding integral polyhedron) using Hartmann's
algorithm~\cite{CHK09}.  However, the number of inequalities needed
can grow exponentially~\cite{Hartmann88}.  If we fix the dimension
$n$, the running time is polynomial in the number of inequalities $m$
and their descriptive size.

\begin{lemma}[Integral polyhedra]\label{lem-poly-integral}
A integral polyhedron contains no integer points if and only if it is
empty.
\end{lemma}
\begin{proof}
If the polyhedron is empty, it cannot contain integer points.
Conversely, assume the integral polyhedron is not empty.  Using the
terminology of the proof of \autoref{lem-farkas-affine}, we know that
the primal problem $P$ is integral and thus has an integer optimal
solution~\cite{Schrijver99}.  This integer optimal solution is an
integer point in the polyhedron.
\end{proof}

According to \autoref{lem-poly-integral}, if we manage to make the
polyhedron in \eqref{eq-motzkin1} integral, we can equivalently
transform an integer universally quantified statement into a rational
existentially quantified one using \hyperref[thm-motzkin]{Motzkin's
Theorem}.  However, this is not applicable for our method because the
polyhedra in \eqref{eq-ii4}, \eqref{eq-ic4}, and \eqref{eq-ti4}
contain free variables.  Making just the stem and loop transitions
integral does preserve more solutions; however, we cannot obtain
completeness by this approach, as the following example illustrates.

\begin{example}\label{ex-non-integral}
Consider the following program $\prog_\mathrm{int}$.
\begin{center}
\begin{minipage}{40mm}
\begin{lstlisting}
assume($2y \geq z$);
while($q \geq 0 \;\land\; z = 1$):
    $q$ := $q - 2y + 1$;
\end{lstlisting}
\end{minipage}
\end{center}
Clearly, $f(q, y, z) = q + 1$ is a ranking function for
$\prog_\mathrm{int}$, hence we use the affine template.  The only
inductive invariant is $2y - z \geq 0$ since inductive invariants have
to be implied by the stem.  These invariants are not sufficient to
prove that $f(q, y, z)$ is indeed a ranking function:
\begin{align*}
&f(q, y, z) - f(q', y', z') - \mu \cdot (2y - z) \\
=\; &q - (q - 2y + 1) - (2y - z) \\
=\; &z - 1
= 0, \text{ but should be positive.}
\end{align*}

The invariant $2y \geq z$ and the loop condition $z = 1$ imply over
the integers that $y \geq 1$ and hence
\begin{align*}
\forall x, x' \in \mathbb{Z}^n.\;
  \loopt(x, x') \land 2y - z \geq 0
  \rightarrow f(x) - f(x') \geq 1
\end{align*}
is valid.  Computing the integral hull of $\loopt \land 2y - z \geq 0$
yields $y \geq 1$ and with this inequality the ranking function can be
discovered.
\end{example}

%% file: 6nl.tex

\clearpage
\chapter{Non-linearity in the Constraints}
\label{ch-nl-constraints}

In this chapter we discuss the constraints generated
in \autoref{sec-constraints}.  We asses the number of variables that
occur in non-linear operations in these constraints using formal
notions introduced in \autoref{sec-nl-dim}.  Furthermore, we give a
theorem that enables us to eliminate some of the Motzkin coefficients
from the constraints in \autoref{sec-nl-constraints}.  We apply this
technique to our ranking function templates in \autoref{sec-nl-rfts}
and give a summarizing overview of the results
in \autoref{sec-rft-overview}.

\section{Definitions}
\label{sec-nl-dim}

\begin{definition}[Dependency graph]
\label{def-rft-dg}
Let $\T$ be a linear ranking function template with variables $D$ and
function symbols $F$.  The template's \emph{dependency graph} is a
graph $G_\T = (D \cup F, E)$ with the set of nodes $D \cup F$ and the
edges
\begin{align*}
E = \{ (f_1, f_2) \in (D \cup F)^2 \mid
  \T \text{ has an atom where both, }
  f_1 \text{ and } f_2 \text{ occur} \}.
\end{align*}
\end{definition}

It follows from the definition that the dependency graph $G_\T$ of a
ranking function template $\T$ is reflexive and symmetric
(undirected).  Given a variable or function symbol $f \in D \cup F$,
we denote by $[f]$ the connected component\footnote{A connected
component is a maximal subset of nodes such that these nodes are
pairwise connected by paths.} that contains $f$.

\begin{example}\label{ex-rft-dg}
The following table lists the set of connected components in the
dependency graph for the ranking function templates introduced
in \autoref{ch-templates}.  See \autoref{fig-rft-dg} for a
visualization.
\begin{center}
\renewcommand{\arraystretch}{1.5} 
\begin{tabular}{ll}
\ref{eq-rft-affine} & $\big\{ \{ f, \delta \} \big\}$ \\
\ref{eq-rft-multiphase}
  & $\big\{ \{ f_i, \delta_i \} \mid 1 \leq i \leq k \big\}$ \\
\ref{eq-rft-pw} & $\big\{ \{ f_1, \ldots, f_k, \delta \} \big\}
  \cup \big\{ \{ g_i \} \mid 1 \leq i \leq k \big\}$ \\
\ref{eq-rft-lex}
  & $\big\{ \{ f_i, \delta_i \} \mid 1 \leq i \leq k \big\}$
\end{tabular}
\end{center}
\end{example}

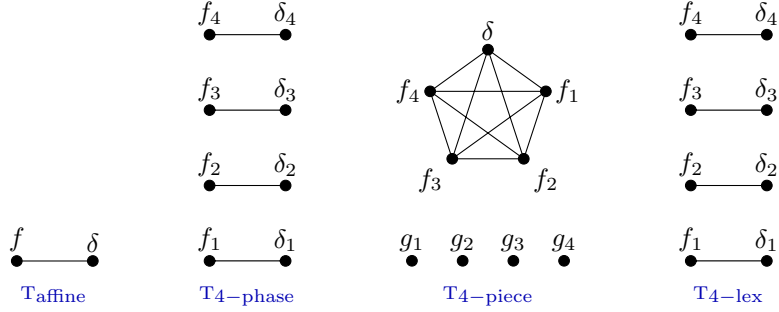
\begin{figure}
\begin{center}
\renewcommand{\tabcolsep}{1.5em}
\begin{tabular}{cccc}
\begin{tikzpicture}
\draw[fill] (0, 0) circle(0.07) node[above] {$f$};
\draw[fill] (1, 0) circle(0.07) node[above] {$\delta$};
\draw (0, 0) -- (1, 0);
\end{tikzpicture}
&
\begin{tikzpicture}
\draw[fill]
  \foreach \x in {1,...,4} {
    (0, \x) circle(0.07) node[above] {$f_\x$}
    (1, \x) circle(0.07) node[above] {$\delta_\x$}
    (0, \x) -- (1, \x)
  };
\end{tikzpicture}
&
\begin{tikzpicture}
\newdimen\R
\R=8mm
\draw[rotate=90] (0:\R)
  \foreach \x in {72,144,...,360} {  -- (\x:\R) }
  \foreach \x in {144,288,...,576} {  -- (\x:\R) }
     -- cycle (360:\R) node[above] {$\delta$}
     -- cycle (288:\R) node[right] {$f_1$}
     -- cycle (216:\R) node[below right] {$f_2$}
     -- cycle (144:\R) node[below left] {$f_3$}
     -- cycle  (72:\R) node[left] {$f_4$};
\draw[rotate=90,fill] (0:\R)
  \foreach \x in {72,144,...,360} { (\x:\R) circle (0.07) };

\draw[fill]
  \foreach \x in {1,...,4} {
    ({(\x-2.5)*0.667}, -2) circle (0.07) node[above] {$g_\x$}
  };
\end{tikzpicture}
&
\begin{tikzpicture}
\draw[fill]
  \foreach \x in {1,...,4} {
    (0, \x) circle(0.07) node[above] {$f_\x$}
    (1, \x) circle(0.07) node[above] {$\delta_\x$}
    (0, \x) -- (1, \x)
  };
\end{tikzpicture} \\
\ref{eq-rft-affine}
& \hyperref[eq-rft-multiphase]{$\T_{4\mathrm{-phase}}$}
& \hyperref[eq-rft-pw]{$\T_{4\mathrm{-piece}}$}
& \hyperref[eq-rft-lex]{$\T_{4\mathrm{-lex}}$}
\end{tabular}
\end{center}
\caption{
The dependency graph of the affine, 4-phase, 4-piece and
4-lexicographic template.  The number of connected components is 1, 4,
5, and 4 respectively.  The graphs' reflexive edges are not shown.
}\label{fig-rft-dg}
\end{figure}

Next, we define colorings, coloring graphs and suitable colorings for
a template $\T$.  A suitable coloring selects the occurrences of atoms
of the template whose Motzkin coefficient we can eliminate in the
constraint \eqref{eq-ii5} $\land$ \eqref{eq-ic5}
$\land$ \eqref{eq-ti5} (see \autoref{sec-nl-constraints}).

\begin{definition}[Coloring]\label{def-coloring}
Let $\T = \bigwedge_{i \in I} \bigvee_{j \in J_i} \T_{i,j}$ be a
linear ranking function template in CNF and let $D$ be the variables
and $F$ be the function symbols of $\T$.  A \emph{coloring $\eta$ of
$\T$} is mapping from occurrences of atoms of $\T$ to
$\{ \cwhite, \cred, \cblue \}$.  The occurrence of an atom $\T_{i,j}$
is called \emph{red} iff it is mapped to $\cred$, \emph{blue} iff it
is mapped to $\cblue$, and uncolored otherwise.
\end{definition}

In \autoref{def-coloring} we consider occurrences of atoms rather than
atoms because an atom may occur multiple times in the same template
and we want to be able to distinguish these occurrences.  For
simplicity, we will sometimes write that an atom is red, blue or
uncolored respectively, if it is clear from context which occurrence
we mean.  Moreover, by stating $f$ occurs in a red atom, we mean $f$
occurs in an atom which has an occurrence that is colored red.

\begin{definition}[Coloring graph]\label{def-coloring-graph}
Let $\T = \bigwedge_{i \in I} \bigvee_{j \in J_i} \T_{i,j}$ be a
linear ranking function template in CNF with variables $D$ and
function symbols $F$, and let $\eta$ be a coloring for $\T$.
The \emph{coloring graph} is a directed graph $G_\eta = (\mathcal{K},
E)$ such that the following holds.
\begin{itemize}
\item The set of nodes $\mathcal{K}$ is the set of connected
  components of $G_\T$.
\item For all $f_1, f_2 \in D \cup F$, there is an edge from the
  connected component of $f_1$ to the connected component of $f_2$,
  i.e., $([f_1], [f_2]) \in E$, if and only if $f_1$ occurs in an red
  atom $\T_{i_1,j_1}$ and $f_2$ occurs in a blue atom $\T_{i_2,j_2}$
  and $i_1 = i_2$, i.e., $\T_{i_1, j_1}$ and $\T_{i_2, j_2}$ occur in
  the same conjunct.
\end{itemize}
\end{definition}

\begin{definition}[Suitable coloring]\label{def-suitable-coloring}
Let $\T$ be a linear ranking function template in CNF and let $F$ be
the function symbols and $D$ be the variables of $\T$.  A coloring
$\eta$ is \emph{suitable for $\T$} iff the following holds.
\begin{enumerate}[a)]
\item Every conjunct of $\T$ contains exactly one red atom.
\item For every $f_1, f_2 \in D \cup F$ that occur in two different
  blue atoms, there is no path between $f_1$ and $f_2$ in the
  dependency graph $G_\T$.
\item The coloring graph $G_\eta$ is acyclic.
\end{enumerate}
\end{definition}

\noindent
See \autoref{fig-rft-colored} on page \pageref{fig-rft-colored} for a
visualization of some suitable colorings for our ranking function
templates.  We also give two detailed examples in the following.

\begin{example}\label{ex-coloring-rft-affine}
Consider the linear template \ref{eq-rft-affine}
from \autoref{def-rft-affine}.  Since \ref{eq-rft-affine} does not
contain any disjunctions in CNF, every conjunct contains exactly one
atom.  Therefore the only suitable coloring for \ref{eq-rft-affine} is
one that colors all occurrences of atoms red according
to \autoref{def-suitable-coloring} (a).  As no atoms are colored blue,
conditions of \autoref{def-suitable-coloring} (b) and (c) are
trivially satisfied.
\end{example}

According to \autoref{def-suitable-coloring} (a), occurrences of atoms
in a conjunct that contains only one atom have to be colored red.

\begin{example}\label{ex-coloring-rft-3phase}
Consider the $3$-phase template.
\begin{align}
\begin{aligned}
&\delta_1 > 0 \;\land\; \delta_2 > 0 \;\land\; \delta_3 > 0 \\
\land\; \big( &f_1(x) > 0 \;\lor\; f_2(x) > 0 \;\lor\; f_3(x) > 0 \big) \\
\land\;\;\; &f_1(x') < f_1(x) - \delta_1 \\
\land\; \big( &f_2(x') < f_2(x) - \delta_2 \;\lor\; f_1(x) > 0 \big) \\
\land\; \big( &f_3(x') < f_3(x) - \delta_3 \;\lor\; f_2(x) > 0
  \;\lor\; f_1(x) > 0\big)
\end{aligned}
\tag{$\T_{3\mathrm{-phase}}$}\label{eq-rft-3phase}
\end{align}
We construct a coloring $\eta$ for \ref{eq-rft-3phase}, given in CNF.
The following atoms have to be colored red according
to \autoref{def-suitable-coloring} (a).
\begin{align*}
\delta_1 > 0, && \delta_2 > 0, && \delta_3 > 0,
&& f_1(x') < f_1(x) - \delta_1.
\end{align*}
The remaining candidates are the atoms for blue coloring are
\begin{align*}
f_1(x) > 0 \;\lor\; f_2(x) > 0 \;\lor\; f_3(x) > 0, \\
f_2(x') < f_2(x) - \delta_2 \;\lor\; f_1(x) > 0, \\
f_3(x') < f_3(x) - \delta_3 \;\lor\; f_2(x) > 0 \;\lor\; f_1(x) > 0.
\end{align*}
Recall that although $f_1(x) > 0$ occurs three times in this list, we
consider it as three different occurrences of the atom
in \ref{eq-rft-3phase}.  By \autoref{ex-rft-dg}, the dependency graph
$G_{\T_{3\mathrm{-phase}}}$ has three connected components.  We color
the following two atoms blue.
\begin{align*}
f_2(x') < f_2(x) - \delta_2 \text{  and  }
f_3(x') < f_3(x) - \delta_3.
\end{align*}
We complete our coloring by choosing the color red for the three
occurrences of the atoms $f_1(x) > 0$.  Note that this choice for
$\eta$ is not the only possibility.  We visualize the coloring $\eta$:
\begin{align*}
&{\color{myred}\delta_1 > 0} \;\land\; {\color{myred}\delta_2 > 0}
\;\land\; {\color{myred}\delta_3 > 0} \\
\land\; \big( &{\color{myred}f_1(x) > 0} \;\lor\; f_2(x) > 0 \;\lor\;
  f_3(x) > 0 \big) \\
\land\;\;\; &{\color{myred}f_1(x') < f_1(x) - \delta_1} \\
\land\; \big( &{\color{myblue}f_2(x') < f_2(x) - \delta_2}
  \;\lor\; {\color{myred}f_1(x) > 0} \big) \\
\land\; \big( &{\color{myblue}f_3(x') < f_3(x) - \delta_3}
  \;\lor\; f_2(x) > 0
  \;\lor\; {\color{myred}f_1(x) > 0} \big)
\end{align*}

Let us check the conditions of \autoref{def-suitable-coloring}:
\begin{enumerate}[a)]
\item We colored exactly one atom in each conjunct red.
\item The two sets of variables and function symbols
  $\{ f_2, \delta_2 \}$ and $\{ f_3, \delta_3 \}$ that occur in blue
  colored atoms are different connected components of the dependency
  graph of \ref{eq-rft-3phase}.
\item The coloring graph $G_\eta$ is acyclic:
\begin{center}
\begin{tikzpicture}
\node (k1) at (0, 0) {$\{ f_1, \delta_1 \}$};
\node (k2) at (2, 0) {$\{ f_2, \delta_2 \}$};
\node (k3) at (4, 0) {$\{ f_3, \delta_3 \}$};
\draw[->,bend right] (k1) to node {} (k2);
\draw[->,bend left] (k1) to node {} (k3);
\end{tikzpicture}
\end{center}
\end{enumerate}
Hence $\eta$ is a suitable coloring for the ranking function
template \ref{eq-rft-3phase}.
\end{example}

\begin{lemma}\label{lem-deg-components}
Let $\T$ be a ranking function template, let $G_\T$ be the dependency
graph of $\T$ and let $\eta$ be a suitable coloring for $\T$.  If
$G_\T$ has $c$ connected components, then the number of occurrences of
atoms colored blue is at most $c - 1$.
\end{lemma}
\begin{proof}
Let $M$ be the function assigning to every blue atom $A$ the connected
component of the variables and function symbols occurring in $A$.  The
function $M$ is injective: if there are two blue atoms $A_1, A_2$ such
that $M(A_1) = M(A_2)$, then from \autoref{def-suitable-coloring} (b)
follows that $A_1 = A_2$.

Assume there are $c$ or more blue atoms, and let $K_1$ be some
connected component in $G_\T$.  Consider the conjunct of $A_1 =
M^{-1}(K_1)$.  By \autoref{def-suitable-coloring} (a), there is a red
atom $A_2$ in this conjunct; let $K_2 = M(A_2)$.  We have that $(K_2,
K_1)$ is an edge in the coloring graph $G_\eta$.

Consequently, every node in the finite coloring graph $G_\eta$ has an
incoming edge, and thus the graph must contain a cycle.  This
contradicts \autoref{def-suitable-coloring} (c).
\end{proof}

\begin{definition}[Degree of a template]\label{def-degree}
Let $\T$ be a linear ranking function template.  We define
the \emph{degree of a coloring $\eta$} of $\T$ as
\begin{align*}
\deg_\T(\eta) = \#\, \eta^{-1}(\{ \cwhite \}),
\end{align*}
the number of occurrences of atoms that are uncolored.
The \emph{degree of $\T$} is the minimal degree of all suitable
colorings of $\T$.
\end{definition}

As we will show in \autoref{sec-nl-constraints}, the degree of $\T$ is
the number of non-linear Motzkin coefficients of atoms of $\T$ that we
are not going to eliminate from our constraints.

\begin{example}\label{ex-def-affine}
By \autoref{ex-coloring-rft-affine}, the only suitable coloring
for \ref{eq-rft-affine} is a coloring $\eta$ such that all atoms are
colored red.  The template \ref{eq-rft-affine} has a degree of at most
\begin{align*}
\deg_{\T_\mathrm{affine}}(\eta) = 0.
\end{align*}
\end{example}

The number of connected components in the dependency graph of a
template $\T$ gives rise to a lower bound on the degree of $\T$
according to \autoref{lem-deg-components}.

\begin{lemma}\label{lem-eta-minimal}
Let $\T$ be a ranking function template, let $G_\T$ be the dependency
graph of $\T$ and let $\eta$ be a suitable coloring for $\T$.  If
$G_\T$ has $c$ connected components and $c - 1$ occurrences of atoms
are colored blue by $\eta$, then the degree of $\T$ is
$\deg_\T(\eta)$.
\end{lemma}
\begin{proof}
By \autoref{def-suitable-coloring} (a), in every coloring the same
number of occurrences of atoms is colored red.  Moreover,
by \autoref{lem-deg-components}, no more than $c - 1$ atoms may be
colored blue, therefore there is no coloring that such that less atoms
are uncolored and thus $\deg_\T(\eta)$ is minimal.
\end{proof}

\begin{example}\label{ex-coloring-rft-3phase2}
Recall the linear ranking function template \ref{eq-rft-3phase} and
its suitable coloring $\eta$ from \autoref{ex-coloring-rft-3phase}.
\begin{align*}
\deg_{\T_{3\mathrm{-phase}}}(\eta) = 3.
\end{align*}
The dependency graph of \ref{eq-rft-3phase} has three components.
By \autoref{lem-eta-minimal} coloring two atoms blue implies that
$\deg_{\T_{3\mathrm{-phase}}}(\eta)$ is minimal and
therefore \ref{eq-rft-3phase} has degree $3$.
\end{example}

\begin{definition}[Non-linear dimension]
\label{def-nl-dim}
Let $\varphi$ be a formula in non-linear arithmetic containing the
free or existentially quantified variables $V$.  The \emph{non-linear
dimension} of $\varphi$ is the size of the smallest subset of
variables $V' \subseteq V$ such that $\varphi$ becomes a formula in
linear arithmetic when assigning a value to each variable in $V'$
(removing their quantifiers).
\end{definition}

A formula $\varphi$ in non-linear arithmetic is also a formula in
linear arithmetic if it uses no non-linear operations (multiplication
of variables).

\begin{example}\label{ex-nl}
Consider the formula in non-linear arithmetic
\begin{align*}
\varphi(y) \equiv \exists x.\; x \cdot y = 1.
\end{align*}
We have one non-linear operation in $\varphi$: the multiplication
$x \cdot y$.  The formula $\varphi$ becomes linear after the
assignment to the variable $x$ or the variable $y$.
\begin{align*}
\varphi_1 &\equiv \exists x.\; x \cdot 3 = 1
  \equiv \exists x.\; x + x + x = 1 \\
\varphi_2(y) &\equiv \frac{1}{2} \cdot y = 1 \equiv y = 2
\end{align*}
Here we used the assignment $x \mapsto \frac{1}{2}$ and $y \mapsto 3$
respectively.  Consequently, $\varphi$ has non-linear dimension $1$.
\end{example}

\section{Non-linear dimension of the Constraints}
\label{sec-nl-constraints}

In this section we examine the non-linear dimension of the constraints
generated in \autoref{sec-constraints}.  First we show that the
constraints \eqref{eq-ii5} $\land$ \eqref{eq-ic5}
$\land$ \eqref{eq-ti5} can be simplified.  We eliminate quantifiers by
fixing the value of the quantified variables: a variable $v$ is fixed
to a finite set of values $\{ u_1, \ldots, u_m \}$ by replacing
$\exists v.\, \varphi(v)$ with
$\varphi(u_1) \lor \ldots \lor \varphi(u_m)$.

\begin{theorem}[Omitting quantifiers]\label{thm-omit-quantifiers}
Let $\eta$ be a suitable coloring for the linear ranking function
template $\T$ according to \autoref{def-suitable-coloring}.  If
for \emph{every} existentially-quantified subformula
in \eqref{eq-ii5}, \eqref{eq-ic5} and \eqref{eq-ti5}, we eliminate
quantifiers by fixing Motzkin coefficients to a finite set of values
as described below, then we obtain \emph{equivalent} constraints.
\begin{enumerate}[I.]
\item The Motzkin coefficient $\xi$ in the invariant initiation
  \eqref{eq-ii5} is fixed to $\{ 0, 1 \}$.
\item The Motzkin coefficient $\chi_2$ in the invariant
  consecution \eqref{eq-ic5} is fixed to $\{ 0, 1 \}$.
\item The Motzkin coefficients $\zeta_j$ in \eqref{eq-ti5} of every
  atom $\T_{i,j}(x, x')$ colored red or blue by $\eta$ is fixed to $\{
  0, 1 \}$.
\item The Motzkin coefficients $\xi_\ell$ in \eqref{eq-ti5} of strict
  invariants are fixed to $\{ 0, 1 \}$, Motzkin coefficients of
  non-strict invariants are fixed to $\{ 1 \}$.
\end{enumerate}
\end{theorem}

\noindent
Applying the quantifier eliminations I, II and IV
from \autoref{thm-omit-quantifiers} to \eqref{eq-ii5}, \eqref{eq-ic5},
and \eqref{eq-ti5} yields the following constraints (elimination III
is omitted for clarity).

\begin{align} 
\begin{aligned}
\bigwedge_{\ell \in L'} \bigwedge_{n \in N} \bigvee_{\xi \in \{0, 1\}}
&\exists \lambda, \mu \geq 0. \\
&\;\quad \tr{\lambda} B_n + \tr{\mu} B_n'
  + \xi \tr{\abovebelow{s_\ell}{0}} = 0 \\
&\land\; \tr{\lambda} b_n + \tr{\mu} b_n' + \xi t_\ell \leq 0 \\
&\land\; \big( \tr{\lambda} b_n < 0
  \;\lor\; \xi + \sum \mu > 0 \big)
\end{aligned}\label{eq-ii6}\tag{II6}
\end{align}
\begin{align} 
\begin{aligned}
\bigwedge_{\ell \in L'} \bigwedge_{m \in M} \bigvee_{\chi_2 \in \{0,1\}}
&\exists \lambda, \chi_1, \mu \geq 0.\; \\
&\;\quad \tr{\lambda} A_m + \tr{\mu} A_m'
  + \chi_2 \tr{\abovebelow{0}{s_\ell}}
  - \chi_1 \tr{\abovebelow{s_\ell}{0}} = 0 \\
&\land\; \tr{\lambda} c_m + \tr{\mu} c_m'
    + (\chi_2 - \chi_1) t_\ell \leq 0 \\
&\land\; \big( \tr{\lambda} c_m - \chi_1 t_\ell < 0
  \;\lor\; \chi_2 + \sum \mu > 0 \big)
\end{aligned}\label{eq-ic6}\tag{IC6}
\end{align}
\begin{align} 
\begin{aligned}
\bigwedge_{i \in I} \bigwedge_{m \in M}
  \bigvee_{(\xi_\ell)_{\ell \in L'} \in \{ 0, 1 \}^{L'}}
&\exists \lambda, (\zeta_j)_{j \in J_i}, \mu \geq 0. \\
&\quad\; \tr{\lambda} A_m + \tr{\mu} A_m'
  + \sum_{\ell \in L} \xi_\ell \tr{\abovebelow{s_{\ell,i,m}}{0}}
  + \sum_{j \in J_i} \zeta_j \tr{d_{i,j}} = 0 \\
&\land\; \tr{\lambda} c_m + \tr{\mu} c_m'
  + \sum_{\ell \in L} \xi_\ell t_{\ell,i,m}
  + \sum_{j \in J_i} \zeta_j e_{i,j} \leq 0 \\
&\land\; \big( \tr{\lambda} c_m
  + \sum_{\ell \in L} \xi_\ell t_{\ell,i,m}
  + \sum_{j \in J_i^\geq} \zeta_j e_{i,j} < 0 \\
&\quad\quad \lor\; \sum_{j \in J_i^>} \zeta_j + \sum \mu > 0 \big)
\end{aligned}\label{eq-ti6}\tag{TI6}
\end{align}

\begin{proof}[Proof of \autoref{thm-omit-quantifiers}]
We need to show that if there is a solution to \eqref{eq-ii5}
$\land$ \eqref{eq-ic5} $\land$ \eqref{eq-ti5}, then there is also a
solution to \eqref{eq-ii6} $\land$ \eqref{eq-ic6}
$\land$ \eqref{eq-ti6}; the converse is clear.  The four variables
fixed to specific values, I--IV, are discussed independently in the
following.
\begin{enumerate}[I.]
\item Let $\lambda$, $\mu$ and $\xi$ be a solution for the
invariant initiation \eqref{eq-ii5}:
\begin{align}
\begin{aligned}
&\tr{\lambda} B_n + \tr{\mu} B_n'
  + \xi \tr{\abovebelow{s_\ell}{0}} = 0 \\
\land\; &\tr{\lambda} b_n + \tr{\mu} b_n' + \xi t_\ell \leq 0 \\
\land\; &\big( \tr{\lambda} b_n < 0 \;\lor\; \xi + \sum \mu > 0 \big)
\end{aligned}
\label{eq-omit-proof-1}
\end{align}
  If $\xi = 0$, there is
  nothing to show.  Otherwise we can pick an assignment to the
  corresponding constraints in \eqref{eq-ii6},
\begin{align}
\begin{aligned}
&\tr{\lambda'} B_n + \tr{\mu'} B_n'
  + \tr{\abovebelow{s_\ell}{0}} = 0 \\
\land\; &\tr{\lambda'} b_n + \tr{\mu'} b_n' + t_\ell \leq 0 \\
\land\; &\big( \tr{\lambda'} b_n < 0
  \;\lor\; 1 + \sum \mu' > 0 \big)
\end{aligned}
\label{eq-omit-proof-2}
\end{align}
  by setting $\lambda' = \frac{\lambda}{\xi}$ and $\mu'
  = \frac{\mu}{\xi}$, since $\lambda$ and $\mu$ only occur in these
  three atoms.  Thus \eqref{eq-omit-proof-1} is satisfiable
  iff \eqref{eq-omit-proof-2} is.

\item Analogously to I, if $\chi_2 \neq 0$, we can divide 
  the solution to the invariant consecution \eqref{eq-ic5} by
  $\chi_2$.

\item Let $\mathcal{K}$ denote the set of connected components in the
  ranking function template's dependency graph $G_\T$.  Furthermore,
  let $G_\eta = (\mathcal{K}, E)$ be the coloring graph according
  to \autoref{def-coloring-graph}.  This graph is finite and acyclic,
  therefore we find a connected component $K \in \mathcal{K}$ that has
  no incoming edges in $G_\eta$.  We will show that fixing the values
  of the Motzkin coefficients of atoms where variables and function
  symbols from $K$ occur is equivalent, and then remove $K$ from the
  coloring graph $G_\eta$.  We do this iteratively for all nodes of the
  coloring graph.  Every atom of $\T$ contains variables or function
  symbols, and thus we cover all described quantifier eliminations.

  First, for every red atom $\T_{i,j^\ast}$ where variables or
  function symbols from $K$ occur, we divide the conjunct's Motzkin
  coefficients by the Motzkin coefficient of $\T_{i,j^\ast}$
  analogously to I and II.

  Second, the variables and function symbols of $K$ occur in at most
  one blue colored atom $\T_{i_0, j_0}$ according
  to \autoref{def-suitable-coloring} (b).  Let $\zeta$ be the Motzkin
  coefficient of $\T_{i_0,j_0}$ in \eqref{eq-ti5} as assigned
  by \hyperref[thm-motzkin]{Motzkin's Theorem} in transformation step
  5.  We assume $\zeta > 0$.

  Let $\T_{i_0, j_\mathrm{red}}$ be the red atom of the conjunct
  $\bigvee_{j \in J_{i_0}} \T_{i_0, j}$ which contains $\T_{i_0, j_0}$
  and let $\zeta_\mathrm{red}$ be the Motzkin coefficient of $\T_{i_0,
  j_\mathrm{red}}$ in \eqref{eq-ti5}.  We know that we already handled
  the $\T_{i_0, j_\mathrm{red}}$ in an earlier step, otherwise there
  would be an incoming edge to $K$ in the coloring graph $G_\eta$.
  Thus we have divided the conjunct by $\zeta_\mathrm{red}$ if
  $\zeta_\mathrm{red} \neq 0$.

  For every function symbol $f \in K$ and every variable $d \in K$, we
  pick $\frac{\zeta}{\zeta_\mathrm{red}} \cdot f$ as a new assignment
  to $f$ and $\frac{\zeta}{\zeta_\mathrm{red}} \cdot d$ as a new
  assignment to $d$ ($\zeta \cdot f$ for $f$ and $\zeta \cdot d$ for
  $d$ in case $\zeta_\mathrm{red} = 0$).  By \autoref{def-linear-rft},
  $\T_{i,j}$ can be written as
\begin{align*}
\sum_{f \in F_{i,j}} \left( \alpha_f f(x) + \beta_f f(x') \right)
  + \sum_{d \in D_{i,j}} \gamma_d d \rhd 0,
\end{align*}
  and this is equivalent to the following by multiplication with
  $\frac{\zeta}{\zeta_\mathrm{red}} > 0$.
\begin{align*}
\sum_{f \in F_{i,j}} \left(
    \alpha_f \Big(\frac{\zeta}{\zeta_\mathrm{red}} \cdot f \Big)(x)
  + \beta_f \Big(\frac{\zeta}{\zeta_\mathrm{red}} \cdot f \Big)(x')
  \right) + \sum_{d \in D_{i,j}}
    \gamma_d \Big(\frac{\zeta}{\zeta_\mathrm{red}} \cdot d\Big)
  \rhd 0
\end{align*}

\item For every inductive invariant $\psi \equiv s^T x + t \rhd 0$,
  the multiple
\begin{align*}
\psi' \equiv (\xi \cdot s^T) x + \xi \cdot t \rhd 0
\end{align*}
  is also an inductive invariant for every $\xi > 0$ if $\psi$ is a
  strict invariant ($\rhd =\; >$) and for all $\xi \geq 0$ if $\psi$
  is a non-strict invariant ($\rhd =\; \geq$).  The variables of every
  invariant $\psi_{\ell,i,m}$ occur only in its
  initiation \eqref{eq-ii5}, its consecution \eqref{eq-ic5} and in
  exactly one template implication in \eqref{eq-ti5} (due to the
  transformation step 3).  Hence for every solution to \eqref{eq-ii5}
  $\land$ \eqref{eq-ic5} we can pick $\frac{s_\ell}{\xi_\ell}$,
  $\frac{t_\ell}{\xi_\ell}$ as a solution to \eqref{eq-ii6}
  and \eqref{eq-ic6}. \qedhere
\end{enumerate}
\end{proof}

\noindent
The proof of \autoref{thm-omit-quantifiers} motivates why we need the
complicated restrictions to suitable colorings
in \autoref{def-suitable-coloring}.  These are the weakest
requirements to a coloring such that we can eliminate the Motzkin
coefficients of the colored atoms.  We get the elimination of red
atoms `for free'---we remove these by dividing the assignments of all
other Motzkin coefficients by this value.  We remove blue atoms by
rescaling the assignment to the template's variables and function
symbols.  However, we have to avoid cyclic dependencies, otherwise the
rescaling operation never terminates.  That is why we need to define
the notion of a coloring graph in \autoref{def-coloring-graph}.

Next, let us assess the non-linear dimension of the
constraints \eqref{eq-ii0} $\land$ \eqref{eq-ic0}
$\land$ \eqref{eq-ti0} before our transformations
in \autoref{sec-constraints}.  We want to compare this to the
non-linear dimension of the constraints \eqref{eq-ii6}
$\land$ \eqref{eq-ic6} $\land$ \eqref{eq-ti6}, the output of our
method.

\begin{theorem}\label{thm-nl-0}
Let $L$ be the index set of invariants, let $F$ and $D$ be the
function symbols and variables of the ranking function template $\T$
respectively.  Let $n$ denote the number of lasso program variables.
The constraints \eqref{eq-ii0} $\land$ \eqref{eq-ic0}
$\land$ \eqref{eq-ti0} have non-linear dimension
\begin{align*}
(n + 1)\#L + (n + 1)\#F + \#D.
\end{align*}
\end{theorem}
\begin{proof}
Non-linear operations occur only with the invariants $\psi(x)$ and the
ranking function template's atoms $\T_{i,j}(x, x')$.  Because we
cannot choose $x$ or $x'$ since they are universally quantified, we
have to choose as the set of variables that occur in non-linear
operations
\begin{align*}
V = \{ s_\ell, t_\ell \mid \ell \in L \} \cup D \cup F.
\end{align*}
The vectors $s_\ell$ and the affine-linear function symbols $f \in F$
have a total number of $n$ and $n + 1$ variables respectively. 
\end{proof}

\begin{theorem}\label{thm-nl-6}
Let $\T$ be a linear ranking function template and let $L' = L \times
I \times M$ as in \autoref{sec-constraints}: $L$ is the index set of
invariants, $M$ is the index set of the loop transition's disjunctions
in CNF and $I$ is the index set of conjunctions in the ranking
function template's DNF.  Let $\eta$ be a suitable coloring for $\T$.
The constraint \eqref{eq-ii6} $\land$ \eqref{eq-ic6}
$\land$ \eqref{eq-ti6} has non-linear dimension at most
\begin{align*}
\#M \cdot \deg_\T(\eta) + \#L'.
\end{align*}
\end{theorem}
\begin{proof}
We first count the number of \hyperref[thm-motzkin]{Motzkin's Theorem}
applications in transformation step 5:
\begin{itemize}
\item $\#L \cdot \#M \cdot \#I \cdot \#N$ from \eqref{eq-ii4},
\item $\#L \cdot (\#M)^2 \cdot \#I$ from \eqref{eq-ic4}, and
\item $\#I \cdot \#M$ from \eqref{eq-ti4}.
\end{itemize}
Let $n_\stemt$ be the total number of inequalities in the stem
transition and $n_\loopt$ be the total number of inequalities in the
loop transition.  Thus the constraint \eqref{eq-ii6}
$\land$ \eqref{eq-ic6} $\land$ \eqref{eq-ti6} has a total number of
\begin{align*}
\#L \cdot \#M \cdot \#I \cdot n_\stemt + \#L \cdot \#M \cdot
\#I \cdot n_\loopt + \#I \cdot n_\loopt
\end{align*}
Motzkin coefficients $\lambda$ and $\mu$, as well as $\#L'$ Motzkin
coefficients $\chi_1$ for invariants $\#M \cdot \deg_\T(\eta)$ Motzkin
coefficients $\zeta$ for the ranking function template.
\end{proof}

In order to eliminate more variables in our constraints, we introduce
non-decreasing invariants~\cite{HHLP13}, as a restricted class of
inductive invariants.

\begin{definition}[Non-decreasing invariant]
\label{def-non-decreasing}
An affine-linear inductive invariant $\psi(x) \equiv \tr{s} x + t \geq
0$ is \emph{non-decreasing} iff
\begin{align*}
\models \forall x, x'.\;
  \loopt(x, x') \rightarrow \tr{s} x' - \tr{s} x \geq 0.
\end{align*}
\end{definition}

Restricting the inductive invariants to non-decreasing invariants is
equivalent to fixing their Motzkin coefficients $\chi_1$ in the
invariant consecution \eqref{eq-ic6} to the value $1$.  This enables
the following corollary to \autoref{thm-nl-6}.

\begin{corollary}\label{cor-nl-6}
Let $\T$, $M$ and $\eta$ as in \autoref{thm-nl-6}.  When using only
non-decreasing invariants, the constraint \eqref{eq-ii6}
$\land$ \eqref{eq-ic6} $\land$ \eqref{eq-ti6} has non-linear dimension
at most $\#M \cdot \deg_\T(\eta)$.
\end{corollary}

\begin{example}\label{ex-non-decreasing1}
The inductive invariant $y \geq 1$ from \autoref{ex-running-3} is in
fact non-decreasing:
\begin{align*}
\models \forall q, y, q', y'.\;
  q \geq 0 \;\land\; q' = q - y \;\land\; y' = y + 1
  \;\rightarrow\; y' - y \geq 0
\end{align*}
\end{example}

\begin{example}\label{ex-non-decreasing2}
Consider the program $\prog_\mathrm{diff42}$~\cite{HHLP13}.
\begin{center}
\begin{minipage}{40mm}
\begin{lstlisting}
$q$ := $y + 42$;
while ($q \geq 0$):
    $y$ := $2 \cdot y - q$;
    $q$ := $(y + q) / 2$;
\end{lstlisting}
\end{minipage}
\end{center}
$\prog_\mathrm{diff42}$ instantiates the affine template with the
ranking function $f(q, y) = q + 1$ and the non-decreasing supporting
invariant $q - y \geq 42$:
\begin{align*}
q' - y' = \frac{y + q}{2} - (2y - q)
  = \frac{3}{2}(q - y)
  \geq \frac{3}{2} \cdot 42
  \geq 42
\end{align*}
\end{example}

Non-decreasing inductive invariants are weaker in expressiveness,
however they still cover a wide range of practical cases.  An inductive
invariant is non-decreasing if
\begin{itemize}
\item the invariant involves only variables that are not
  modified by the loop transition, or
\item the invariant is an equality.
\end{itemize}

\begin{example}\label{ex-non-decreasing3}
The variable $\chi_1$ in the invariant consecution \eqref{eq-ic5}
cannot be generally restricted to any finite set of values analogously
to \autoref{def-non-decreasing} without losing solutions.  Let $\alpha
> 1$ be some fixed constant and consider the following linear lasso
program $\prog_\alpha$:
\begin{center}
\begin{minipage}{50mm}
\begin{lstlisting}
assume($y := \alpha$);
while ($q \geq 0$):
    $q$ := $q - y$;
    $y$ := $\frac{1}{\alpha}(y + \alpha - 1)$;
\end{lstlisting}
\end{minipage}
\end{center}
The program $\prog_\alpha$ terminates, because $y \geq 1$ is an
invariant of $\prog_\alpha$.  Furthermore, the invariant is inductive:
the stem $y = \alpha$ implies $y \geq 1$, and
\begin{align*}
y' = \frac{y + \alpha - 1}{\alpha}
   = \frac{y}{\alpha} + \frac{\alpha - 1}{\alpha}
   \geq \frac{1}{\alpha} + \frac{\alpha - 1}{\alpha}
   = 1.
\end{align*}
However, $y \geq 1$ is not a non-decreasing invariant:
\begin{align}
y' - y = \frac{y + \alpha - 1}{\alpha} - y
       = \frac{\alpha - 1}{\alpha} \cdot (1 - y).
\label{eq-ex-non-decreasing3}
\end{align}
We are stuck, because \eqref{eq-ex-non-decreasing3} cannot be inferred
to be non-negative: there is no lower bound on $-y$.  We have to
choose $\chi_1 = \frac{1}{\alpha}$ in \eqref{eq-ic6} because
\begin{align*}
y' - \frac{1}{\alpha} \cdot y
  = \frac{y + \alpha - 1}{\alpha} - \frac{1}{\alpha} \cdot y
  = \frac{\alpha - 1}{\alpha}
  \geq 0.
\end{align*}
In fact, every value $\chi_1 > \frac{1}{\alpha}$ will not work for the
same reason as in \eqref{eq-ex-non-decreasing3}.
\end{example}

\begin{corollary}\label{cor-linear}
If the linear ranking function template $\T$ has degree $\leq 0$ and
we consider only non-decreasing invariants, then the
constraint \eqref{eq-ii6} $\land$ \eqref{eq-ic6}
$\land$ \eqref{eq-ti6} is linear.
\end{corollary}
\begin{proof}
According to \autoref{cor-nl-6}, the constraints have non-linear
dimension at most $\#M \cdot \deg_\T(\eta)$ for a suitable coloring
$\eta$ of $\T$.  Consequently, for $\deg_\T(\eta) = 0$, we have no
non-linear operations.  Since there is no universal quantification,
the generated constraint is linear by \autoref{def-degree}.
\end{proof}

The constraints generated by \ref{eq-rft-affine} are due Podelski and
Rybalchenko~\cite{PR04}.  In \cite{HHLP13} this template is extended
to incorporate non-decreasing inductive invariants, and the generated
template coincides with the one generated here, if we restrict
ourselves to non-decreasing invariants and conjunctive linear lasso
programs.

\section{Application to our Templates}
\label{sec-nl-rfts}

In this section we assess the degree of the ranking function templates
introduced in \autoref{ch-templates}.  See \autoref{fig-rft-colored}
on page \pageref{fig-rft-colored} for a visualization.

\begin{figure}[ht]
\begin{center}
{\setlength{\tabcolsep}{1.2em}
\begin{tikzpicture}[scale=1.5]
\def\radius{0.07}
\foreach \y in {-0.2, 0.5} {
  \draw (0, \y) -- (2, \y);
}
\draw (0.8, 1.5) -- (2, 1.5);
\draw (-0.2, 1) -- (2, 1);

\def\x{0}
\draw[dotted] (\x, -0.2) to node[left] {$[g_1]$} (\x, 0.5);
\draw[dotted] (\x, 0.5) .. controls (\x, 1) and ({\x-0.2}, 0.5)
  .. ({\x-0.2}, 1);
\draw[dotted] (\x, 0.5) .. controls (\x, 1) and ({\x+0.2}, 0.5)
  .. ({\x+0.2}, 1);
\draw[fill=mywhite] (\x, -0.2) circle (\radius);
\draw[fill=mywhite] (\x, 0.5) circle (\radius);
\draw[fill=myblue] ({\x-0.2}, 1.0) circle (\radius);
\draw[fill=mywhite] ({\x+0.2}, 1.0) circle (\radius);

\def\x{1}
\draw[dotted] (\x, -0.2) to node[left] {$[g_2]$} (\x, 0.5) -- (\x, 1);
\draw[dotted] (\x, 1) .. controls (\x, 1.5) and ({\x-0.2}, 1)
  .. ({\x-0.2}, 1.5);
\draw[dotted] (\x, 1) .. controls (\x, 1.5) and ({\x+0.2}, 1)
  .. ({\x+0.2}, 1.5);
\draw[fill=mywhite] (\x, -0.2) circle (\radius);
\draw[fill=mywhite] (\x, 0.5) circle (\radius);
\draw[fill=myblue] ({\x-0.2}, 1.5) circle (\radius);
\draw[fill=mywhite] ({\x+0.2}, 1.5) circle (\radius);

\def\x{2}
\draw[dotted] (\x, -0.7) -- (\x, -0.2) to node[left] {$[\delta]$}
  (\x, 0.5) -- (\x, 2.5);
\draw[fill=myred] \foreach \y in {-0.7, -0.2, 0.5, 1, 1.5, 2, 2.5} {
  (\x, \y) circle (\radius)
};
\end{tikzpicture} \\
\hyperref[eq-rft-pw]{$\T_{2\mathrm{-piece}}$}
 \\[3em]
\begin{tabular}{ccc}
\begin{tikzpicture}[scale=1.5]
\def\radius{0.07}
\draw[dotted] (0, 0.0) to node[left] {$[f]$} (0, 0.75) -- (0, 1.25);
\draw[fill=myred] (0, 0.0) circle (\radius);
\draw[fill=myred] (0, 0.75) circle (\radius);
\draw[fill=myred] (0, 1.25) circle (\radius);
\end{tikzpicture}
&
\begin{tikzpicture}[scale=1.5]
\def\radius{0.07};
\def\spacing{0.7};
\draw (\spacing, 0) -- (\spacing*4, 0);
\foreach \i in {1,...,4} {
  \def\y{0.5*\i}
  \draw (\spacing, \y) -- ({\spacing*(5 - \i)}, \y);
}

\foreach \i in {1,...,4} {
  \def\x{\spacing*\i}
  \draw[dotted] (\x, 0.0) -- (\x, 2.0)
    to node[left] {$[f_\i]$} (\x, 2.8);
  \draw[fill=myred] (\x, 2.8) circle (\radius);

  \foreach \j in {\i,...,5} {
    \def\y{2.5-0.5*\j}
    \ifnum\i=1
      \draw[fill=myred] (\x, \y) circle (\radius);
    \else
      \ifnum\i=\j
        \draw[fill=myblue] (\x, \y) circle (\radius);
      \else
        \draw[fill=mywhite] (\x, \y) circle (\radius);
      \fi
    \fi
  };
};
\end{tikzpicture}
&
\begin{tikzpicture}[scale=1.5]
\def\radius{0.07};
\def\spacing{0.7};
\foreach \i in {1,...,4} {
  \def\y{0.5*\i}
  \draw (\spacing, \y) -- ({\spacing*(5 - \i)}, \y);
}

\foreach \i in {1,...,4} {
  \def\x{\spacing*\i}
  \draw[dotted] (\x, 0.0) -- (\x, 2.0)
    to node[left] {$[f_\i]$} (\x, 2.8);
  \draw[fill=myred] (\x, 2.8) circle (\radius);

  \foreach \j in {\i,...,4} {
    \def\y{2.5-0.5*\j}
    \ifnum\i=1
      \draw[fill=myred] (\x, \y) circle (\radius);
    \else
      \ifnum\i=\j
        \draw[fill=myblue] (\x, \y) circle (\radius);
      \else
        \draw[fill=mywhite] (\x, \y) circle (\radius);
      \fi
    \fi
  };
  \draw[fill=myred] (\x, 0) circle (\radius);
};
\end{tikzpicture} \\[1em]
\ref{eq-rft-affine}
& \hyperref[eq-rft-multiphase]{$\T_{4\mathrm{-phase}}$}
& \hyperref[eq-rft-lex]{$\T_{4\mathrm{-lex}}$}
\end{tabular}}
\end{center}
\caption{
Visualization of suitable colorings for the $2$-piece, affine,
$4$-phase and $4$-lexicographic ranking function template.  Every
circle represents an occurrence of an atom; its color is determined by
$\eta$.  Two circles are connected with a solid line if they occur
together in one conjunct.  Every connected component of the coloring
graph is represented by a dotted line that connects all atoms where
these variable and function symbols occur.
We can now easily check the conditions
of \autoref{def-suitable-coloring}: every solid line connects exactly
one red circle (a), every dotted line connects at most one blue circle
(b).  We can extract the coloring graph by drawing a directed edge
between two connected components if there is a solid line connecting a
red circle with a blue one (c).
The degree of the templates is the number of white circles (uncolored
atoms).
}
\label{fig-rft-colored}
\end{figure}
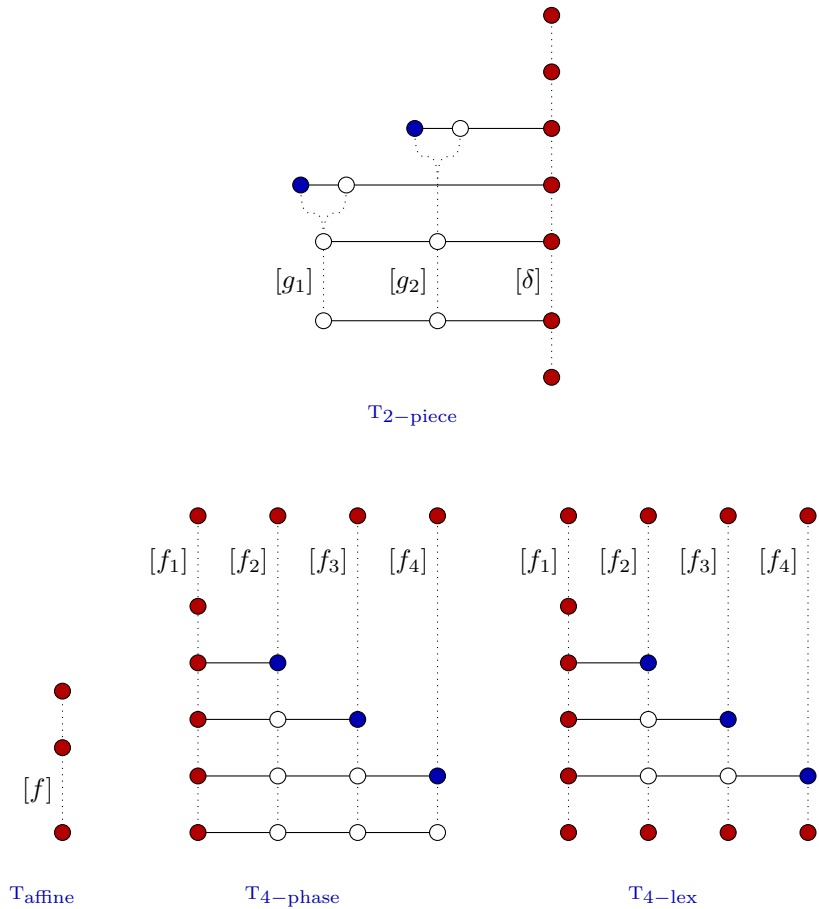

\begin{lemma}\label{lem-rft-multiphase-deg}
The ranking function template \ref{eq-rft-multiphase} has degree
$\frac{1}{2}k(k-1)$.
\end{lemma}
\begin{proof}
Consider the following coloring $\eta$.
\begin{align*}
&\bigwedge_{i=1}^k {\color{myred} \delta_i > 0} \\
\land\; &\Big( {\color{myred} f_1(x) > 0}
  \;\lor\; \bigvee_{i=1}^k f_i(x) > 0 \Big) \\
\land\; &{\color{myred} f_1(x') < f_1(x) - \delta_1} \\
\land\; &\bigwedge_{i=2}^k \Big(
  {\color{myblue}f_i(x') < f_i(x) - \delta_i} \;\lor\;
  {\color{myred} f_1(x) > 0} \;\lor\;
  \bigvee_{j=2}^{i-1} f_j(x) > 0 \Big)
\end{align*}
We check the requirements of \autoref{def-suitable-coloring}:
\begin{enumerate}[a)]
\item Every conjunct contains exactly one red atom.
\item The sets $\{ f_i, \delta_i \}$ are the connected components
  of \ref{eq-rft-multiphase} according to \autoref{ex-rft-dg}.
\item The coloring graph $G_\eta$ has the edges $([f_1], [f_i])$ for
  $i = 2, \ldots, k$.  Hence it is acyclic.
\end{enumerate}
We conclude that $\eta$ is a suitable coloring
of \ref{eq-rft-multiphase}.
\begin{align*}
\deg_{\T_{k\mathrm{-phase}}}(\eta)
  = \frac{k(k+5)}{2} - (k - 1) - (2k + 1)
  = \frac{k(k-1)}{2}.
\end{align*}
The number of blue colored atoms is $k - 1$ and the number of
connected components is $k$, therefore, by \autoref{lem-eta-minimal},
this degree is minimal.
\end{proof}

\begin{example}\label{ex-3phase-3}
Consider the 2-phase template defined in \autoref{def-rft-multiphase}.
It has degree $\frac{1}{2} \cdot 2(2 - 1) = 1$.  When considering only
non-decreasing supporting invariants, \autoref{cor-nl-6} states that
the generated constraints have non-linear dimension at most $\#M$,
where $\#M$ is number of disjunctions in the normal form of the linear
lasso program's loop transition.  If we build constraints for the
lasso program $\prog_{2\mathrm{-phase}}$ from \autoref{ex-2phase-1},
we have only one non-linear variable according to \autoref{thm-nl-6}
when considering non-decreasing supporting invariants.
\end{example}

\begin{lemma}\label{lem-rft-pw-deg}
The ranking function template \ref{eq-rft-pw} has degree $2k^2 - 1$.
\end{lemma}
\begin{proof}
Consider the following coloring $\eta$.
\begin{align*}
&{\color{myred}\delta > 0} \\
\land\; &\bigwedge_{i=1}^k \Big( {\color{myblue} g_i(x) < 0} \;\lor\;
  g_i(x') < 0 \;\lor\; {\color{myred} f_i(x') < f_i(x) - \delta} \Big) \\
\land\; &\bigwedge_{i=1}^k \bigwedge_{j \neq i}
  \Big( g_i(x) < 0 \;\lor\; g_j(x') < 0
  \;\lor\; {\color{myred} f_j(x') < f_i(x) - \delta} \Big) \\
\land\; &\bigwedge_{i=1}^k {\color{myred} f_i(x) > 0} \\
\land\; &\Big( {\color{myred} g_1(x) \geq 0}
  \;\lor\; \bigvee_{i=2}^k g_i(x) \geq 0 \Big)
\end{align*}
We check the requirements of \autoref{def-suitable-coloring}:
\begin{enumerate}[a)]
\item Every conjunct contains exactly one red atom.
\item The sets $\{ g_i \}$ and $\{ \delta, f_1, \ldots, f_k \}$ are the
  connected components of \ref{eq-rft-pw} according
  to \autoref{ex-rft-dg}.
\item The coloring graph $G_\eta$ has the edges $([\delta], [g_i])$
  for all $i$ and $([g_i], [g_1])$ for $i > 1$. Hence it is acyclic.
\end{enumerate}
We conclude that $\eta$ is a suitable coloring of \ref{eq-rft-pw}.
\begin{align*}
\deg_{\T_{k\mathrm{-piece}}}(\eta)
  = (3k^2 + 2k + 1) - k - (k^2 + k + 2)
  = 2k^2 - 1.
\end{align*}
The number of blue colored atoms is $k$ and the number of
connected components is $k + 1$, therefore, by \autoref{lem-eta-minimal},
this degree is minimal.
\end{proof}

\begin{lemma}\label{lem-rft-lex-deg}
The ranking function template \ref{eq-rft-lex} has degree
$\frac{1}{2}(k - 1)(k - 2)$.
\end{lemma}
\begin{proof}
This proof is similar to the proof of \autoref{lem-rft-multiphase-deg}
since the lexicographic termination has the same dependency graph as
the multiphase template.  Consider the following coloring $\eta$.
\begin{align*}
&\bigwedge_{i=1}^{k} {\color{myred} \delta_i > 0} \\
\land\; &\bigwedge_{i=1}^k {\color{myred} f_i(x) > 0} \\
\land\; &{\color{myred} f_1(x') \leq f_1(x)} \\
\land\; &\bigwedge_{i=2}^{k-1} \Big( {\color{myblue} f_i(x') \leq f_i(x)}
  \;\lor\; {\color{myred} f_1(x') < f_1(x) - \delta_1}
  \;\lor\; \bigvee_{j=2}^{i-1} f_j(x') < f_j(x) - \delta_j \Big) \\
\land\; &\Big( {\color{myblue} f_k(x') < f_k(x) - \delta_k}
  \;\lor\; {\color{myred} f_1(x') < f_1(x) - \delta_1}
  \;\lor\; \bigvee_{i=2}^{k-1} f_i(x') < f_i(x) - \delta_i
  \Big)
\end{align*}
We check the requirements of \autoref{def-suitable-coloring}:
\begin{enumerate}[a)]
\item Every conjunct contains exactly one red atom.
\item $\{ f_i, \delta_i \}$ are the
  connected components of \ref{eq-rft-lex} according
  to \autoref{ex-rft-dg}.
\item The coloring graph $G_\eta$ has the edges $([f_i], [f_1])$ for
  all $i > 1$.  Hence it is acyclic.
\end{enumerate}
We conclude that $\eta$ is a suitable coloring of \ref{eq-rft-lex}.
\begin{align*}
\deg_{\T_{k\mathrm{-lex}}}(\eta)
  = \frac{k(k + 5)}{2} - (k - 1) - 3k
  = \frac{(k - 1)(k - 2)}{2}.
\end{align*}
The number of blue colored atoms is $k - 1$ and the number of
connected components is $k$, therefore, by \autoref{lem-eta-minimal},
this degree is minimal.
\end{proof}

\begin{example}\label{ex-rft-lex-2}
The program $\prog_\mathrm{gcd}$ from \autoref{ex-rft-lex}
instantiates the 2-lexicographic template.  The two invariants
$y_1 \geq 1$ and $y_2 \geq 1$ are non-decreasing.  The coloring $\eta$
from \autoref{lem-rft-lex-deg} has degree
\begin{align*}
\deg_{\T_{2\mathrm{-lex}}}(\eta) = \frac{1}{2} \cdot 1 \cdot 0 = 0.
\end{align*}
Consequently, by \autoref{cor-linear}, the generated constraints are
linear if we are considering non-decreasing invariants, and otherwise
\begin{align*}
\#L' = \#L \cdot \#I \cdot \#M = 1 \cdot 6 \cdot 2 = 12
\end{align*}
by \autoref{thm-nl-6}.
\end{example}

\section{Overview}
\label{sec-rft-overview}

Linear ranking function templates use affine-linear function variables
to synthesize a termination argument.  When constructing a ranking
function from the assignment to these function variables, the ordinal
ranking equivalents (\autoref{def-ordinal-ranking}) of these linear
functions turn out to be central components.  The image of the ranking
functions constructed from the ordinal ranking equivalents is an
ordinal, namely the \emph{ranking structure} of $\T$.

The following table gives an overview of the presented ranking
function templates: the affine template, the $k$-phase template, the
$k$-piece template and the $k$-lexicographic template.  We state the
number of conjuncts and atoms when written in CNF, the number of
connected components in their dependency graph, their ranking
structure and degree (as proven in \autoref{sec-nl-rfts}).
\begin{center}
\renewcommand{\arraystretch}{1.5} 
\begin{tabular}{lcccc}
& \ref{eq-rft-affine} & \ref{eq-rft-multiphase} & \ref{eq-rft-pw}
& \ref{eq-rft-lex} \\
\hline
Conjuncts & $3$ & $2k + 1$ & $k^2 + k + 2$ & $3k$ \\
Atoms & $3$ & $\frac{1}{2}k(k + 5)$ & $3k^2 + 2k + 1$ & $\frac{1}{2}k(k + 5)$ \\
Connected comp. & $1$ & $k$ & $k + 1$ & $k$ \\
Ranking structure & $\omega$ & $\omega \cdot k$ & $\omega$ & $\omega^k$ \\
Degree & $0$ & $\frac{1}{2}k(k - 1)$ & $2k^2 - 1$ &
$\frac{1}{2}(k - 1)(k - 2)$
\end{tabular}
\end{center}

%% file: 7solving.tex

\clearpage
\chapter{Solving the Constraints}
\label{ch-solving}

In \autoref{ch-constraints} we constructed constraints that are
satisfiable for a given linear lasso program only if there exists a
ranking function of a specialized form.  These constraints are of the
existential fragment of non-linear real arithmetic.  In this chapter
we want to discuss strategies available for solving these constraints.

Since Tarski published the first decision procedure for the first
order theory of the reals~\cite{Tarski51}, several other algorithms
have been proposed (an overview can be found in Grant Passmore's PhD
thesis~\cite{Passmore11}).  The Grigor’ev--Vorobjov--Theorem states
that, in theory, the existential fragment of non-linear real
arithmetic can be solved in single exponential time~\cite{GV88}.
However, this bound seems to be only of limited practical
relevance \cite{Hong91}.
\emph{Cylindrical algebraic decomposition} (CAD) is most successful
in practice, despite its doubly exponential worst case complexity
bound.  This is still an active area of research and recently
significant progress has been achieved in terms of running
time~\cite{JM12}.

For Farkas' Lemma based constraints, specialized solving algorithms
have been thought out, e.g. under-approximation using
heuristics~\cite{CSS03,SSM04} and bisection search combined with
linear constraint solving~\cite{BMS05linrank}.  However, we found it
is both feasible and practical to utilize an off-the-shelf SMT solver
to find a solution to small and medium-sized examples.  Nonetheless,
in this section we want to examine the CAD algorithm for our case in
detail and discuss some possible runtime mitigation.  The goal is to
motivate that while the generated constraints are indeed non-linear,
for most practical cases they are still not vastly more difficult to
solve than a linear constraint.

After a brief introduction to cylindrical algebraic decomposition
in \autoref{sec-cad-intro}, we exemplary solve the constraints
corresponding to one invariant consecution in \autoref{sec-cad}.  We
show that invariants that depend only on a constant number of loop
inequalities correspond to solutions of the non-linear system that can
be discovered in polynomial time.

\section{Introduction to CAD}
\label{sec-cad-intro}

Cylindrical algebraic decomposition was first presented by
Collins~\cite{Collins75}.  We give a brief introduction by example
based on Jirstrand's technical report~\cite{Jirstrand95}.

The goal is to find a partition of $\mathbb{R}^n$ such that the given
set of polynomials have constant sign on each component.  These
components are finitely represented by single points; satisfiability
of the non-linear constraints can be decided by checking these
representative points.

\begin{example}\label{ex-cad-intro}
Consider the following system of polynomial equations in the two
variables $x$ and $y$.
\begin{align}
\begin{aligned}
x^2 + y^2 - 2 &< 0 \\
x^3 - y^2 &= 0
\end{aligned}
\label{eq-cad-intro}
\end{align}
The two involved polynomials are $p_1(x, y) = x^2 + y^2 - 2$ and
$p_2(x, y) = x^3 - y^2$.
\end{example}

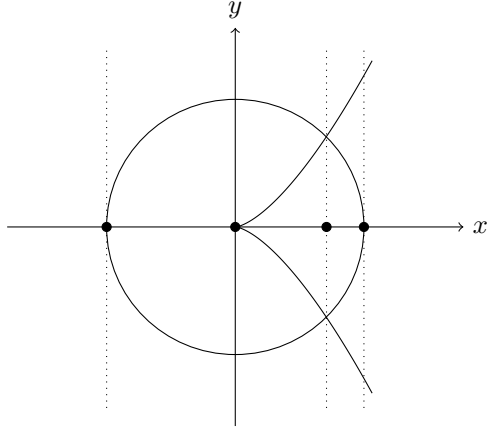
\begin{figure}[ht]
\begin{center}
\begin{tikzpicture}[scale=1.2]
\draw[->] (-2.5, 0) -- (2.5,0) node[right] {$x$};
\draw[->] (0, -2.2) -- (0,2.2) node[above] {$y$};

\draw[domain=0:1.5,samples=50] plot (\x,{sqrt(\x*\x*\x)});
\draw[domain=0:1.5,samples=50] plot (\x,{-sqrt(\x*\x*\x)});

\draw (0,0) circle (1.41);

\draw[dotted] (1,-2) -- (1, 2);
\draw[dotted] (1.41,-2) -- (1.41,2);
\draw[dotted] (-1.41,-2) -- (-1.41,2);
\draw[fill] (-1.41, 0) circle (0.05);
\draw[fill] (0, 0) circle (0.05);
\draw[fill] (1, 0) circle (0.05);
\draw[fill] (1.41, 0) circle (0.05);
\end{tikzpicture}
\end{center}
\caption{The zero sets of the polynomials in \eqref{eq-cad-intro} and
some projections on the $x$-axis (solid dots).
}\label{fig-cad-intro}
\end{figure}

\noindent
The algorithm works in three phases:
\begin{enumerate}
\item \emph{The projection phase:} all points of zero sets
corresponding to vertical tangents, singularities and intersections
are projected to the lower dimension eliminating one variable.
\item \emph{The base phase:} For mono-variant polynomials, all roots
can be enumerated and we thus get a sign-invariant decomposition of
$\mathbb{R}^1$.
\item \emph{The extension phase:} The technique of the base phase is
applied recursively to lift the sign invariant decomposition from
$\mathbb{R}^i$ to $\mathbb{R}^{i+1}$.
\end{enumerate}
In \autoref{ex-cad-intro}, projection to the $x$-axis yields the solid
dots as depicted in \autoref{fig-cad-intro}.  These are the two points
$(-\sqrt{2}, 0)$ and $(\sqrt{2}, 0)$ corresponding to vertical
tangents of $p_1$, as well as $(0, 0)$, the singularity of $p_2$, and
finally $(1, 0)$, the projection of the intersection points $(1, 1)$
and $(1, -1)$ of $p_1$ and $p_2$.
This gives us the decomposition on the $x$-axis defined by these four
points as well as the intervals inbetween them.  We can evaluate the
signs of the polynomials $p_1$ and $p_2$ on these regions:
\begin{center}
\small
\setlength{\tabcolsep}{0.4em}
\begin{tabular}{r|ccccccccc}
$x$ & $(-\infty, -\sqrt{2})$ & $-\sqrt{2}$ & $(-\sqrt{2}, 0)$ & $0$
  & $(0, 1)$ & $1$ & $(1, \sqrt{2})$ & $\sqrt{2}$
  & $(\sqrt{2}, \infty)$ \\
\hline
$\mathrm{sign}(p_1(x))$ & $+$ & $0$ & $-$ & $-$ & $-$ & $-$ & $-$
  & $0$ & $+$ \\
$\mathrm{sign}(p_2(x))$ & $-$ & $-$ & $-$ & $0$ & $+$ & $+$ & $+$
  & $+$ & $+$
\end{tabular}
\end{center}
Over each region, we can calculate recursively the decomposition of
$\mathbb{R}^2$ and check for points that satisfy the system of
inequalities \eqref{eq-cad-intro}; in our example, the point $(0, 0)$
is a possible solution.

For the precise definition of the projection phase, along with the
required notion of \emph{principal subresultant coefficients}
($\mathrm{psc}$) used in the proof of \autoref{lem-cad-form}, we
reference Jirstrand's report~\cite{Jirstrand95} as this would go far
beyond the scope of this work.

\section{Solving with CAD}
\label{sec-cad}

As an illustration we discuss the CAD of the invariant consecution.
In theory the arguments offered here also apply to the complete
constraints, although admittedly additional difficulties arise due to
the simultaneous presence of more than one variable that occurs in
non-linear operations.

We start with a single invariant consecution \eqref{eq-ic4}:
\begin{align}
\forall x, x'.\; \psi_\ell(x) \land \loopt_m(x, x')
  \rightarrow \psi_\ell(x')
\label{eq-cad1}
\end{align}
For clarity, we drop the indices $m$ of $\loopt$ and $\ell$ of
$\psi$.  Recall that we write
\begin{align*}
\loopt(x, x') &\equiv Ax \leq b \;\land\; A'x < b', \\
\psi(x) &\equiv \tr{s} x + t \geq 0.
\end{align*}
After applying \hyperref[thm-motzkin]{Motzkin's Theorem} in
transformation step 5, we get \eqref{eq-ic5} and fixing the value of
$\chi_2$ according to \autoref{thm-omit-quantifiers}, the
constraints \eqref{eq-ic6} corresponding to \eqref{eq-cad1} are
\begin{align}
\begin{aligned}
\exists \lambda, \chi, \mu \geq 0.\;
&\tr{\lambda} A + \tr{\mu} A' + \tr{\abovebelow{0}{s}}
  - \chi\tr{\abovebelow{s}{0}} = 0 \\
&\land\; \tr{\lambda} c + \tr{\mu} c' + (1 - \chi) t \leq 0 \\
&\land\; (\tr{\lambda} c - \chi t < 0 \;\lor\; \sum \mu > 0)
\end{aligned}\label{eq-cad2}
\end{align}
We are solving for the vectors $\lambda, \mu$, $s \in \mathbb{K}^n$,
and the variables $\chi \in \mathbb{K}$ and $t \in \mathbb{K}$.  The
matrices $A, A'$ and vectors $c, c'$ are constant.  Hence the only
non-linear terms are $\chi \tr{\abovebelow{s}{0}}$ and $\chi t$.  If
we find an assignment for $\chi$, the constraints become linear and
thus can be solved using an SMT solver for linear arithmetic, which
have polynomial runtime complexity~\cite{Schrijver99}.

For simplification, we focus on the first disjunct in \eqref{eq-cad2}
(classical case), since the other case is structurally very similar.
Let $A = (a_{i,j})$, $\lambda = \tr{(\lambda_1 \ldots \lambda_m)},$
and $s = \tr{(s_1 \ldots s_n)}$.  We write \eqref{eq-cad2} explicitly:
\begin{align}
\chi &\geq 0
\label{eq-cad-k0} \\
\lambda_i &\geq 0, &\text{for } 1 \leq i \leq m
\label{eq-cad-k1} \\
\sum_{i=1}^m a_{i,j} \lambda_i - \chi s_j &= 0, &\text{for } 1 \leq j \leq n
\label{eq-cad-k2} \\
\sum_{i=1}^m a_{i,n+j} \lambda_i + s_j &= 0, &\text{for } 1 \leq j \leq n
\label{eq-cad-k3} \\
\sum_{i=1}^m c_i \lambda_i - (1 - \chi) t &< 0
\label{eq-cad-k4}
\end{align}
We solve \eqref{eq-cad-k3} for $s_j$ and eliminate $s_j$
from \eqref{eq-cad-k2} yielding the following system of equations.
\begin{align}
\chi &\geq 0 \label{eq-cad-l0} \\
\lambda_i &\geq 0, &\text{for } 1 \leq i \leq m \label{eq-cad-l1} \\
\sum_{i=1}^m (a_{i,j} + a_{i,n+j} \chi) \lambda_i &= 0,
&\text{for } 1 \leq j \leq n
\label{eq-cad-l2} \\
\sum_{i=1}^m c_i \lambda_i - (1 - \chi) t &\leq 0
\label{eq-cad-l4}
\end{align}

We ignore the constraint \eqref{eq-cad-l4} because in the case where
$\chi \neq 1$, we can always assign $t$ such that this inequality
holds.  For the projection we choose the ordering
$\lambda_1, \ldots, \lambda_m, \chi$.  The set of polynomials for the
CAD is
\begin{align}
P = \{ \chi, \lambda_i \mid 0 \leq i \leq m \}
\cup \Big\{ \sum_{i=1}^m  (a_{i,n+j} \chi + a_{i,j}) \lambda_i
  \mid 0 \leq j \leq n \Big\}.
\label{eq-cad-poly}
\end{align}
In particular, the coefficient polynomials to $\lambda_i$
in \eqref{eq-cad-poly} are linear in $\chi$.

\begin{lemma}\label{lem-cad-form}
In every projection step $k$ in the CAD of \eqref{eq-cad-poly}, the
set of polynomials
\begin{align}
P_k = \{ \lambda_i \mid k \leq i \leq m \}
\;\cup\; \Big\{ \sum_{i=k}^m p_{i,j}(\chi) \lambda_i
  \mid j \in J_k \Big\}
\;\cup\; \{ q_\ell(\chi) \mid \ell \in L_k \}
\label{eq-cad-form}
\end{align}
for suitable index sets $J$ and $L$.  The polynomials $p_{i,j}$ and
$q_\ell$ only involve the variable $\chi$.
\end{lemma}
\begin{proof}
We proceed inductively.  Clearly \eqref{eq-cad-poly} satisfies this
criterion.  Consider the projection of $\lambda_k$.
\begin{itemize}
\item For every polynomial $p \in P_k$, we take $p(\lambda_k = 0)$.
  This yields $\lambda_{k+1}, \ldots, \lambda_m$, $\sum_{i=k+1}^m
  p_{i,j}(\chi) \lambda_i$ and preserves $q_\ell(\chi)$.
\item For every polynomial $p \in P_k$, we calculate
  $\mathrm{psc}_{\lambda_k}(p, \frac{\partial p}{\partial\lambda_k})$.
  Since $p$ is linear in $\lambda_k$, this yields the same results as
  the previous step.
\item For every pair of polynomials $p_{j_1}, p_{j_2} \in P_k$, we calculate
\begin{align*}
    &\mathrm{psc}_{\lambda_k}(p_1, p_2) \\
=\; &\mathrm{psc}_{\lambda_k}\big( \sum_{i=k}^m p_{i,j_1}(\chi) \lambda_i,
  \sum_{i=k}^m p_{i,j_2}(\chi) \lambda_i \big) \\
=\; &\sum_{i=k+1}^m \Big(
  \frac{\mathrm{lcm}(p_{k,j_1}, p_{k, j_2})}{p_{k,j_1}} p_{i,j_1}(\chi)
  - \frac{\mathrm{lcm}(p_{k,j_1}, p_{k, j_2})}{p_{k,j_2}} p_{i,j_2}(\chi)
  \Big) \lambda_i.
\end{align*}
  This is again of the form given in \eqref{eq-cad-form}. \qedhere
\end{itemize}
\end{proof}

According to \autoref{lem-cad-form}, after $m$ projection steps we are
left with a set of polynomials $P_m = \{ q_\ell(\chi) \mid \ell \in
L_m \}$ dependent only on the variable $\chi$.

\begin{lemma}\label{lem-cad-result}
The following holds for $P_m$.
\begin{enumerate}[I.]
\item The degree of any $q_\ell \in P_m$ is at most $2^m$.
\item $\#P_m \leq {n^2}^m$.
\item The number of distinct roots of all $q_\ell \in P_m$ is bounded
  by $2^m n^{2^m}$.
\end{enumerate}
\end{lemma}
\begin{proof} ~
\begin{enumerate}[I.]
\item Initially, all polynomials are linear. In every projection step,
  the maximum degree of polynomials can at most double as the least
  common multiple's degree is less or equal to the degree of the
  product.
\item Since each polynomial is connected with every other polynomial,
  the number of distinct new polynomials is at most squared in every
  projection step.
\item Follows directly from I and II. \qedhere
\end{enumerate}
\end{proof}

\noindent
More importantly, the number of projection steps depend on the number
of $\lambda$-variables $m$.  \hyperref[thm-motzkin]{Motzkin's Theorem}
introduces one $\lambda$-variable for every inequality
in \eqref{eq-motzkin1}.  Therefore the runtime scales with the number
of statements relevant to prove the invariant consecution.  We assume
that in practice, only a small (maybe even constant) number of
inequalities is required to prove an invariant.  This greatly reduces
the bound in \autoref{lem-cad-result} III.

\begin{theorem}\label{thm-constraints-polynomial}
Invariants that depend only on a constant number of loop inequalities
can be discovered in polynomial time.
\end{theorem}
\begin{proof}
Let $e$ be the bound on the required loop inequalities.  There are
$\binom{m}{e} \leq m^e$ possibilities to select $e$ of the $m$
inequalities.  We search for a solution to \eqref{eq-cad1} by applying
CAD to \eqref{eq-cad-l1} and \eqref{eq-cad-l2}.  The result is
described by \autoref{lem-cad-form} and \autoref{lem-cad-result}
states the bound $2^e n^{2^e}$ for distinct values of $\chi$.  Given
possible assignments to $\chi$, we plug every one of these into the
constraints \eqref{eq-cad2} and solve using a solver for linear
arithmetic.  Satisfiability for linear arithmetic is decidable in
polynomial time~\cite{Schrijver99} and we only have polynomially many
values for $\chi$ to try since $e$ is constant.
\end{proof}

Although \autoref{thm-constraints-polynomial} gives a polynomial
algorithm for solving the constraints, it is not practical.  It would
be a great deal more efficient to follow the CAD algorithm in
constructing the solution, which has been omitted here for simplicity
of presentation.  Additionally, there is no good reason not to enlarge
$e$ to take more loop inequalities into consideration until a
predefined time limit runs out.

%% file: 8conclusion.tex

\clearpage
\chapter{Conclusion}
\label{ch-conclusion}

The scope of this work is a new method for synthesizing termination
arguments for linear lasso programs.  This method generalizes existing
methods and extends them in various ways.
In \autoref{sec-related-work} we elaborated on how our method relates
to existing research.

We introduced the notion of \emph{ranking function templates}
in \autoref{ch-templates} and discussed the affine, multiphase,
piecewise and lexicographic templates in detail.  For the
affine-linear functions used in the ranking function templates, we
introduced the notion of \emph{ordinal ranking equivalent} in order to
naturally build ranking functions using ordinal arithmetic from
assignments for the template's variables and functions symbols.

The \emph{multiphase ranking function} is a novel type of ranking
function and received some more detailed investigation.  We showed
that there are conjunctive linear lasso programs that do not have a
multiphase ranking function (\autoref{ex-multiphase-rotation}) and we
showed that the existence of a multiphase ranking function does not
entail information about the program's complexity
(\autoref{ex-multiphase-complexity}).

Notable formal results in this work are the undecidability proof for
termination of linear lasso programs (\autoref{thm-halting}) and the
theorem regarding the removal of quantifiers in our constraints
(\autoref{thm-omit-quantifiers}).

Other contributions include the soundness and completeness statements
for our method (\autoref{thm-soundness}
and \autoref{thm-completeness}), the discourse about the treatment of
mixed integer variable domains (\autoref{sec-integers}), the
assessment of the non-linear dimension before and after our
transformations to the constraints
(\autoref{thm-nl-0}, \autoref{thm-nl-6},
\autoref{cor-nl-6} and \autoref{cor-linear}) and the motivation why
solving the resulting constraints is not necessarily very difficult
(\autoref{sec-cad}).  An overview over the ranking function templates
we consider and their properties can be found
in \autoref{sec-rft-overview}, including their non-linear dimension
and their ordinal ranking structure.

\section{Future Work}

For future work, it would be interesting to see new ranking function
templates.  There certainly are more types of ranking functions that
can be formalized by ranking function templates.  One could
investigate the use of affine-linear, multiphase, piecewise and
lexicographic templates as a `construction kit'.  For example, they
could be combined to more general templates by replacing single
affine-linear functions in the lexicographic template by a piecewise
or multiphase ranking function.

Furthermore, it seems vital to implement our method and try it on real
world examples.  Only experimental evaluation will tell which ranking
function templates are both computationally feasible and practically
relevant.  Ideally, our method would be used in conjunction with a
tool that is able to combine termination arguments for lasso programs
to a termination argument for a complex program.

Moreover, the selection of a ranking function template is not part of
our method.  A heuristic could be devised that intelligently suggests
a template by looking at the program code.  This could reduce required
human interaction and/or speed up the termination argument synthesis.

Our method is not complete on integer lasso programs as discussed
in \autoref{sec-integers}.  Possibly there is a way of making the
polyhedra integral even though they contain free variables.
Otherwise, a different approach for integers needs to be developed.
As integer variables are extremely common in real life code, this
topic requests further attention.

The complexity of our method is centrally determined by the complexity
of non-linear algebraic constraint solving.  Any progress being made
in this field improves the applicability of our method.  Non-linear
constraint solving is an active area of research and recent
progress~\cite{JM12} suggests that algorithmic improvements are not
yet exhausted.

Another plot line unfinished is the decidability of the termination of
conjunctive linear lasso programs.  We conjectured
in \autoref{conj-lasso} that this is decidable, but a proof remains
due.

\subsection*{Acknowledgements}
I would like to thank my supervisor Matthias Heizmann for all his patience and support.
The dialogue with him and his continued encouragement was invaluable to this work.
Furthermore, I extend my gratitude to Fabian Reiter and Pascal Raiola for their very helpful corrections and suggestions.

%% file: errata.tex

\newpage
\section*{Errata}

This is an error corrected version of the original thesis, updated last on \today.
The most important changes are:

\begin{itemize}
\item Removed Lemma 3.13 because its statement was false. Thanks go to Amir Ben-Amram for pointing this out.
\item Added missing $\xi_\ell$ in \autoref{eq-ti6}.
\item Fixed step size in \autoref{ex-running-6}.
\item Corrected the coloring of \ref{eq-rft-lex} in \autoref{lem-rft-lex-deg}.
\item Fixed Typos
\end{itemize}